\pdfoutput=1
\documentclass{article}

    \usepackage[preprint]{neurips_2023}

\usepackage{authblk}
\usepackage[utf8]{inputenc} % allow utf-8 input
\usepackage[T1]{fontenc}    % use 8-bit T1 fonts
\usepackage{hyperref}       % hyperlinks
\usepackage{url}            % simple URL typesetting
\usepackage{booktabs}       % professional-quality tables
\usepackage{amsfonts}       % blackboard math symbols
\usepackage{nicefrac}       % compact symbols for 1/2, etc.
\usepackage{microtype}      % microtypography
\usepackage{xcolor}         % colors

\usepackage{subcaption}
\usepackage{amsmath}
\usepackage[nameinlink,capitalize]{cleveref}
\usepackage{lipsum}		% Can be removed after putting your text content
\usepackage{graphicx}
\usepackage{doi}
\usepackage{dsfont}
\usepackage{ntheorem}
\newtheorem{theorem}{Theorem}
\usepackage[inkscapelatex=false]{svg}
\usepackage[normalem]{ulem}
\usepackage{adjustbox}

\useunder{\uline}{\ul}{}
\newtheorem*{proof}{Proof}
\usepackage{url}

\title{An Efficient Membership Inference Attack for the Diffusion Model by Proximal Initialization}

\author{%
   \textbf{Fei Kong\textsuperscript{1} \quad Jinhao Duan\textsuperscript{2} \quad RuiPeng Ma\textsuperscript{1} \quad Hengtao Shen\textsuperscript{1} \quad Xiaofeng Zhu\textsuperscript{1} \vspace{-5pt}}\\
  \textbf{Xiaoshuang Shi\textsuperscript{1}\thanks{Equal corresponding author} \quad Kaidi Xu\textsuperscript{2\thinspace$*$}} \vspace{3pt}\\
  \textsuperscript{1}University of Electronic Science and Technology of China \\
  \textsuperscript{2}Drexel University 
  \vspace{3pt}\\
  \texttt {kong13661@outlook.com \quad xsshi2013@gmail.com \quad kx46@drexel.edu}

}
\begin{document}

\maketitle

\begin{abstract}
% 	\lipsum[1]

\setcounter{footnote}{0} 

Recently, diffusion models have achieved remarkable success in generating tasks, including image and audio generation. However, like other generative models, diffusion models are prone to privacy issues. In this paper, we propose an efficient query-based membership inference attack (MIA), namely Proximal Initialization Attack (PIA), which utilizes groundtruth trajectory obtained by $\epsilon$ initialized in $t=0$ and predicted point to infer memberships. Experimental results indicate that the proposed method can achieve competitive performance with only two queries on both discrete-time and continuous-time diffusion models. Moreover, previous works on the privacy of diffusion models have focused on vision tasks without considering audio tasks. Therefore, we also explore the robustness of diffusion models to MIA in the text-to-speech (TTS) task, which is an audio generation task. To the best of our knowledge, this work is the first to study the robustness of diffusion models to MIA in the TTS task. Experimental results indicate that models with mel-spectrogram (image-like) output are vulnerable to MIA, while models with audio output are relatively robust to MIA. {Code is available at \url{https://github.com/kong13661/PIA}}.

% 
% In this paper, we investigate the vulnerability of diffusion models to Membership Inference Attack (MIA), a common private concern. 
% We considered the concurrent attack methods Naive Attack and SecMI. SecMI achieve a better performance with the cost of more queries. 

% However, as the model size increases, the cost of time has become an increasingly important factor to consider. To reduce the time cost, we propose Proximal Initialization Attack (PIA), a black attack method that infers memberships by applying $\epsilon$ initialized in $t=0$ to training loss. To evaluate the effectiveness of our method, we tested it on an audio generation model and an image generation model. The results showed that we achieved similar performance to SecMI with only 2 queries, which is one more query than Naive Attack. 

% We evaluated the robustness of TTS diffusion model by applying Naive Attack, SecMI and proposed method to three speech models, and our results showed that in the TTS diffusion model, the model with mel spectrogram output (which can be considered as a type of image output to some extent) was more vulnerable. However, models with direct audio output exhibited better robustness.

% Several works have explored the security and privacy risks in image generation based diffusion generation models. However, there have been no such explorations in the field of audio generation. 

\end{abstract}

% keywords can be removed
% \keywords{First keyword \and Second keyword \and More}

% \section{Introduction}
% \lipsum[2]
% \lipsum[3]

\section{Introduction}

% In recent years, an advanced method called diffusion model \cite{ho2020denoising, Score_Based,Estimating_Gradients} is rising up in the field of generative tasks. It has achieved significant success in various tasks, including image generation \cite{stable_diffusion,Imagen}, audio generation \cite{Popov2021Grad-TTS,Kong2021DiffWave}, video generation \cite{video_diffusion,video_diffusion2}, etc.
% However, as a generative model, the diffusion model may also suffer from privacy risks \cite{private_risk} and copyright disputes \cite{copyright} similar to those faced by other generative models, such as in GANs \cite{gans}/VAEs\cite{vaes}, for instance, Privacy leaking \cite{logan} and data reconstruction \cite{reconstruction}. Recently, there have been some works to explore this topic \cite{SecMI,matsumoto2023membership,hu2023membership}. These works indicate that diffusion models also face privacy issues. 

Recently, the diffusion model \cite{ho2020denoising, Score_Based,Estimating_Gradients} has emerged as a powerful approach in the field of generative tasks, achieving notable success in image generation \cite{stable_diffusion,Imagen}, audio generation \cite{Popov2021Grad-TTS,Kong2021DiffWave}, video generation \cite{video_diffusion,video_diffusion2}, and other domains. However, like other generative models such as GANs \cite{gans} and VAEs \cite{vaes}, the diffusion model may also be exposed to privacy risks \cite{private_risk} and copyright disputes \cite{copyright}. Dangers such as privacy leaks \cite{logan} and data reconstruction \cite{reconstruction} may compromise the model. Recently, some researchers have explored this topic \cite{SecMI,matsumoto2023membership,hu2023membership,carlini2023extracting}, demonstrating that diffusion models are also vulnerable to privacy issues.

Membership Inference Attacks (MIAs) are the most common privacy risks \cite{shokri2017membership}. MIAs can cause privacy concerns directly and can also contribute to privacy issues indirectly as part of data reconstruction. Given a pre-trained model, MIA aims to determine whether a sample is in the training set or not.

Generally speaking, MIA relies on the assumption that a model fits the training data better \cite{yeom2018privacy,shokri2017membership}, resulting in a smaller training loss. Recently, several MIA techniques have been proposed for diffusion models \cite{SecMI,matsumoto2023membership,hu2023membership}. We refer to the query-based methods proposed in \cite{matsumoto2023membership,hu2023membership} as Naive Attacks because they directly employ the training loss for the attack. However, unlike GANs or VAEs, the training loss for diffusion models is not deterministic because it requires the generation of Gaussian noise. The random Gaussian noise may not be the one in which diffusion model fits best. This can negatively impact the performance of the MIA attack. To address this issue, the concurrent work SecMI \cite{SecMI} adopts an iterative approach to obtain the deterministic $\boldsymbol x$ at a specific time $t$, but this requires more queries, resulting in longer attack times. As models grow larger, the time required for the attack also increases, making time an important metric to consider.

To reduce the time consumption, inspired by DDIM and SecMI, we proposed a Proximal Initialization Attack (PIA) method, which derives its name from the fact that we utilize the diffusion model's output at time $t=0$ as the noise $\epsilon$. PIA is a query-based MIA that relies solely on the inference results and can be applied not only to discrete time diffusion models \cite{ho2020denoising,stable_diffusion} but also to continuous time diffusion models \cite{Score_Based}. We evaluate the effectiveness of our method on three image datasets, CIFAR10 \cite{CIFAR10}, CIFAR100 and TinyImageNet for DDPM and on two images dataset, COCO2017 \cite{COCO} and Laion5B \cite{schuhmann2022laion5b} for Stable DIffuion, as well as three audio datasets, LJSpeech \cite{ljspeech17}, VCTK \cite{vctk}, and LibriTTS \cite{zen2019libritts}.

 % To our knowledge, recent works about MIA only focus on images filed, and there has been no research on the privacy of diffusion models in the audio domain. However, since audio, such as a piece of music, faces copyright and privacy issues similar to those in the image domain \cite{cnn,wapo}, it is necessary to conduct privacy research in the audio domain. It is necessary to investigate whether audio is also vulnerable to attacks and what types of diffusion models are more robust to privacy attacks. To verify the robustness of MIA on audio, we  conduct experiment about a Naive Attack, SecMI \cite{SecMI} and our method on Grad-TTS \cite{Popov2021Grad-TTS}, DiffWave \cite{Kong2021DiffWave}, FastDiff \cite{Huang2022FastDiff}.  The results indicate that the robustness of MIA on audio depends on the output type of the model.

To our knowledge, recent research on MIA of diffusion models has only focused on image data, and there has been no exploration of diffusion models in the audio domain. However, audio, such as music, encounters similar copyright and privacy concerns as those in the image domain \cite{cnn,wapo}. Therefore, it is essential to conduct privacy research in the audio domain to determine whether audio data is also vulnerable to attacks and to identify which types of diffusion models are more robust against privacy attacks. To investigate the robustness of MIA on audio data, we conduct experiments using Naive Attack, SecMI \cite{SecMI}, and our proposed method on three audio models: Grad-TTS \cite{Popov2021Grad-TTS}, DiffWave \cite{Kong2021DiffWave}, and FastDiff \cite{Huang2022FastDiff}. The results suggest that the robustness of MIA on audio depends on the output type of the model.

Our contributions can be summarized as follows:

\begin{itemize}
\item We propose a query-based MIA method called PIA. Our method employs the output at $t=0$ as the initial noise and the errors between the forward and backward processes as the attack metric. We generalize the PIA on both discrete-time and continuous-time diffusion models.
\item Our study is the first to evaluate the robustness of MIA on audio data. We evaluate the robustness of MIA on three TTS models (Grad-TTS, DiffWave, FastDiff) and three TTS datasets (LJSpeech, VCTK, Libritts) using Naive Attack, SecMI, and our proposed method.
\item Our experimental results demonstrate that PIA achieves similar AUC performance and higher TPR @ 1\% FPR performance compared to SecMI while being 5-10 times faster, with only one additional query compared to Naive Attack. Additionally, our results suggest that, for text-to-speech audio tasks, models that output audio have higher robustness against MIA attacks than those that output mel-spectrograms, which are the image-like output. Based on our findings, we recommend using generation models that output audio to reduce privacy risks in audio generation tasks.
\end{itemize}

\begin{figure}
 \centering 
\includegraphics[width=\textwidth]{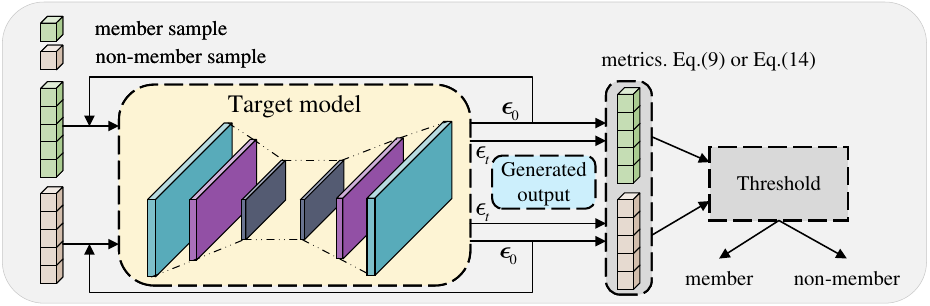} % 1st subfigure: \includegraphics{fig}...
\caption{A overview of PIA. First, a sample is an input into the target model to generate $\boldsymbol\epsilon$ at time $0$. Next, we combine the original sample with $\boldsymbol\epsilon_0$ and input them into the target model to generate $\boldsymbol\epsilon$ at time $t$. After that, we input all three variables into a metric and use a threshold to determine if the sample belongs to the training set.}
        \label{overview}
\end{figure}

\section{Related Works and Background}

\textbf{Generative Diffusion Models} ~ Generative diffusion models have recently achieved significant success in both image \cite{Ramesh2022Hierarchical,stable_diffusion} and audio generation tasks \cite{Huang2022FastDiff,Chen2021WaveGrad,Popov2021Grad-TTS}. Unlike GANs \cite{gans,yuan2020attribute,yuan2023dde}, which consist of a generator and a discriminator, diffusion models generate samples by fitting the inverse process of a diffusion process from Gaussian noise. Compared to GANs, diffusion models typically produce higher quality samples and avoid issues such as checkerboard artifacts \cite{salimans2016improved, Donahue2017Adversarial,Dumoulin2017Adversarially}. A diffusion process is defined as 
$\boldsymbol{x}_{t}=\sqrt{\alpha_{t}}\boldsymbol{x}_{t-1}+\sqrt{\beta_{t}} \boldsymbol\epsilon_{t},  \boldsymbol\epsilon_{t} \sim \mathcal{N}(\mathbf{0}, \mathbf{I}) $, where $\alpha_t+\beta_t=1$ and $\beta_t$ increases gradually as $t$ increases, so that eventually, $\boldsymbol x_t$ approximates a random Gaussian noise. In the reverse diffusion process, $\boldsymbol{x}'_t$ still follows a Gaussian distribution, assuming the variance remains the same as in the forward diffusion process, and the mean is defined as $\tilde{\boldsymbol \mu}_{t}=\frac{1}{\sqrt{a_{t}}}\left(\boldsymbol x_{t}-\frac{\beta_{t}}{\sqrt{1-\bar{a}_{t}}} \bar{\boldsymbol{\epsilon}}_{\theta}(\boldsymbol x_t,t)\right)$, where $\bar\alpha_t=\prod_{k=0}^t \alpha_k$ and $\bar\alpha_t+\bar\beta_t=1$. The reverse diffusion process becomes $\boldsymbol x_{t-1}=\tilde{\boldsymbol \mu}_{t}+\sqrt{\beta_t}\boldsymbol\epsilon,\boldsymbol\epsilon\sim \mathcal{N}(\mathbf{0}, \mathbf{I})$. One can obtain a loss function \cref{loss_func} by minimizing the distance between the predicted and groundtruth distributions. \cite{Score_Based} transforms the discrete-time diffusion process into a continuous-time process and uses SDE ( Stochastic Differential Equation) to express the diffusion process. To accelerate the generation process, several methods have been proposed, such as \cite{Salimans2022Progressive,Dockhorn2022Score-Based,Xiao2022Tackling}. DDIM \cite{DDIM} is another popular method that proposes a forward process different from diffusion process with the same loss function as DDPM, allowing it to reuse the model trained by DDPM while achieving higher generation speed.
\begin{equation}
    L = \mathbb{E}_{x_{0}, \bar{\boldsymbol \epsilon}_{t}}\left[\left\|\bar{\boldsymbol \epsilon}_{t}-\boldsymbol \epsilon_{\theta}\left(\sqrt{\bar{\alpha}_{t}} x_{0}+\sqrt{1-\bar{\alpha}_{t}} \bar{\boldsymbol \epsilon}_{t}, t\right)\right\|^{2}\right].
    \label{loss_func}
\end{equation}

\textbf{Membership Inference Privacy} ~ Different from conventional adversarial attacks~\cite{xu2018structured,xu2020adversarial,zhang2022branch}, Membership inference attack (MIA)~\cite{shokri2017membership} aims to determine whether a sample is part of the training data. It can be formally described as follows: given two sets, the training set $\mathcal D_t$ and the hold-out set $\mathcal{D}_h$, a target model $m$, and a sample $\boldsymbol x$ that either belongs to $\mathcal D_t$ or $\mathcal D_h$, the goal of MIA is to find a classifier or function $f(\boldsymbol x, m)$ that determines which set $\boldsymbol x$ belongs to, with $f(\boldsymbol x, m)\in \{0,1\}$ and $f(\boldsymbol x, m)=1$ indicating that $\boldsymbol x\in \mathcal{D}_t$ and $f(\boldsymbol x, m)=0$ indicating that $\boldsymbol x\in \mathcal{D}_h$. If a membership inference attack method utilizes a model's output obtained through queries to attack the model, it is called query-based attack\cite{SecMI,matsumoto2023membership,hu2023membership}. Typically, MIA is based on the assumption that training data has a smaller loss compared to hold-out data. MIA for generation tasks, such as GANs \cite{logan} and VAEs \cite{hilprecht2019monte,chen2020gan}, has also been extensively researched.

% Depending on the level of access to the model, MIA can be classified into black-box \cite{shokri2017membership,Salem2019ML-Leaks,yeom2018privacy,sablayrolles2019white,Hui2021Practical} and white-box \cite{nasr2019comprehensive,rezaei2020towards} attacks. A MIA is black-box, if we can only access the part or all output of a model $m(\boldsymbol x)$. If we can access more information about the model such as model weights, the MIA is called white-box attack \cite{choquette2021label}. 

Recently, several MIA methods designed for diffusion models have been proposed. \cite{matsumoto2023membership} proposed a method that directly employs the training loss \cref{loss_func} and find a specific $t$ with maximum distinguishability. Because they directly use the training loss, we refer to this method as Naive Attack. SecMI \cite{SecMI} improves the attack effectiveness by iteratively computing the $t$-error, which is the error between the DDIM sampling process and the inverse sampling process at a certain moment $t$.

\textbf{Threat model} ~ We follow the same threat model as \cite{SecMI}, which needs to access intermediate outputs of diffusion models. This is a query-based attack without the knowledge of model parameters but not fully end-to-end black-box. In scenarios such as inpainting \cite{lugmayr2022repaint}, and classification \cite{li2023your}, they also employ the intermediate output of the diffusion model. These works utilize a pre-trained model on a huge dataset to do other tasks, such as inpainting, and classification without fine-tuning. To meet these requirements, future service providers might consider opening up APIs for intermediate outputs. Our work is applicable to such scenarios.

% We will later demonstrate that this is equivalent to computing the value of $\epsilon$ used in the loss calculation.

\section{Methodology}
\label{Methodology}

In this section, we introduce DDIM, a variant of DDPM, and provide a proof that if we know any two points in the DDIM framework, $\boldsymbol x_k$ and $\boldsymbol x_0$, we can determine any other point $\boldsymbol x_t$. We then propose a new MIA method that utilizes this property to efficiently obtain $\boldsymbol x_{t-t'}$ and its corresponding predicted sample $x'_{t-t'}$. We compute the difference between these two points and use it to determine if a sample is in the training set. Specifically, samples with small differences are more likely to belong to the training set. An overview of this proposed method is shown in \cref{overview}.

% we derived our method step-by-step from DDIM. 

% Our proposed method is inspired by both DDIM and method []. Finally, we will introduce our proposed method.

\subsection{Preliminary}

% \textbf{Membership Inference Attack} (MIA) aims to determine whether a sample comes from the training set. As mentioned in the section on related work, given a training set $\mathcal{D}_t$, a hold-out set $\mathcal{D}_h$, a sample $x$, a targeted model $m$, we want to find a function $f(\boldsymbol x, m)\in \{0,1\}$. $f(\boldsymbol x, m)=1$ indicates $x\in \mathcal{D}_t$, otherwise, $x\in \mathcal{D}_h$. In general, the function $f$ can take the following form:

% \begin{equation}
%     f(\boldsymbol x, m)=\mathds{1}[\mathds{P}(x\in \mathcal{D}_t|m)\ge \tau]
%     \label{f}
% \end{equation}

% In \cref{f}, $\mathds{1}[A]$ is an indicator function with $\mathds 1[\text{true}]=1$, $\mathds 1[\text{false}]=0$ and $\tau$ is a constant. 

% \textbf{Denoising Diffusion Implicit Models} (DDIM) aims to accelerate the inference process of diffusion model. DDIM accelerates inference by defining a new diffusion process sharing the same loss as DDPM. Unlike DDPM process, which adds noise from $x_0$ to $x_T$. DDIM defines a diffusion process from $x_T$ to $x_1$ by using $x_0$. \cref{ddim1} and \cref{ddim_x} describe this process. The distribution $q_{\sigma}\left(\boldsymbol{x}_{T} \mid \boldsymbol{x}_{0}\right)$ is the same as DDPM.

\textbf{Denoising Diffusion Implicit Models} To accelerate the inference process of diffusion models, DDIM defines a new process that shares the same loss function as DDPM. Unlike the DDPM process, which adds noise from $x_0$ to $x_T$, DDIM defines a diffusion process from $x_T$ to $x_1$ by using $x_0$. The process is described in \cref{ddim1} and \cref{ddim_x}. The distribution $q_{\sigma}\left(\boldsymbol{x}_T \mid \boldsymbol{x}_0\right)$ is the same as in DDPM.
\begin{equation}
q_{\sigma}\left(\boldsymbol{x}_{1: T} \mid \boldsymbol{x}_{0}\right):=q_{\sigma}\left(\boldsymbol{x}_{T} \mid \boldsymbol{x}_{0}\right) \prod_{t=2}^{T} q_{\sigma}\left(\boldsymbol{x}_{t-1} \mid \boldsymbol{x}_{t}, \boldsymbol{x}_{0}\right),
\label{ddim1}
\end{equation}
\begin{equation}
q_{\sigma}\left(\boldsymbol{x}_{t-1} \mid \boldsymbol{x}_{t}, \boldsymbol{x}_{0}\right)=\mathcal{N}\left(\sqrt{\bar{\alpha}_{t-1}} x_{0}+\sqrt{1-\bar{\alpha}_{t-1}-\sigma_{t}^{2}} \cdot \frac{\boldsymbol{x}_{t}-\sqrt{\bar{\alpha}_{t}} x_{0}}{\sqrt{1-\bar{\alpha}_{t}}}, \sigma_{t}^{2} \boldsymbol{I}\right) .
\label{ddim_x}
\end{equation}

The denoising process defined by DDIM is described below:
\begin{equation}
    \begin{aligned}\left.p(\boldsymbol{x}_{t'} \mid \boldsymbol{x}_{t}\right) & = p\left(\boldsymbol{x}_{t'} \mid \boldsymbol{x}_{t}, \boldsymbol{x}_{0}=\overline{\boldsymbol{\mu}}\left(\boldsymbol{x}_{t}\right)\right) \\ & =\mathcal{N}\left(\boldsymbol{x}_{t'} ;    \frac{\sqrt{\bar\alpha_{t'}}}{\sqrt{\bar\alpha_t}}\left(\boldsymbol{x}_{t}-\left(
    \sqrt{1-\bar\alpha_t}-\frac{\sqrt{\bar\alpha_t}}{\sqrt{\bar\alpha_{t'}}} \sqrt{ 1-\bar\alpha_{t'} -\sigma_{t}^{2}}\right) \boldsymbol{\epsilon}_{\boldsymbol{\theta}}\left(\boldsymbol{x}_{t}, t\right)\right), \sigma_{t}^{2} \boldsymbol{I}\right)
    \end{aligned} 
    \label{ddim_denoise}
\end{equation}

% \subsection{Previous methods}

% Recent works, including our method, are based on the hypothesis: the samples used in training episode have lower loss compared with hold-out ones. So we can utilize some some information to estimate the loss, and mark the samples with lower loss training samples.

% \textbf{Naive attack} Recently, some works in MIA on diffusion model has been presented. We call the method used in [] naive attack for they use the original loss function to attack the targeted method. As descriped in \cref{loss_func}, in every attack, they select a time $t$ and generate a complete random noise to estimate the loss. 

% \textbf{aaa} Inspired by DDIM, aaa first uses forward DDIM to estimate a sequence $\{\boldsymbol{x}_0\dots\boldsymbol{x}_t\}$. Second, the use backward DDIM to estimate the samples $\boldsymbol{x}'_{t-1}$ from $\boldsymbol{x}_t$. Last, the use $\Vert \boldsymbol{x}_t-\boldsymbol{x}'_t \Vert$ to estimate the error between forward sample and backward sample. They call the error \emph{t-error}.

\subsection{Finding Groundtruth Trajectory}
\label{groundtruth_trajectory}
In this section, we will first demonstrate that if we know $\boldsymbol x_k$ and $\boldsymbol x_0$, we can determine any other $\boldsymbol x_t$. Then, we will provide the method for obtaining $\boldsymbol x_k$.

\begin{theorem}
The trajectory of $\{\boldsymbol x_t\}$ is determined if we know $x_0$ and any other point $x_k$ when $\sigma_t=0$ under DDIM framework.
\end{theorem}

\begin{proof}
    In DDIM definition, if standard deviation $\sigma_t=0$, the process adding noise becomes determined. So \cref{ddim_x} can be rewritten to \cref{ddim_determined}.
    \begin{equation}
        \boldsymbol{x}_{t-1}=\sqrt{\bar{\alpha}_{t-1}} \boldsymbol{x}_{0}+\sqrt{1-\bar{\alpha}_{t-1}} \cdot \frac{\boldsymbol{x}_{t}-\sqrt{\bar{\alpha}_{t}} \boldsymbol{x}_{0}}{\sqrt{1-\bar{\alpha}_{t}}}.
        \label{ddim_determined}
    \end{equation}
    Assuming that we know any point $\boldsymbol x_k$. \cref{ddim_determined} can be rewritten as $\frac{\boldsymbol x_{t-1}-\sqrt{\bar{\alpha}_{t-1}} \boldsymbol{x}_{0}}{\sqrt{1-\bar{\alpha}_{t-1}}}=\frac{\boldsymbol{x}_{t}-\sqrt{\bar{\alpha}_{t}} \boldsymbol{x}_{0}}{\sqrt{1-\bar{\alpha}_{t}}}$. By applying this equation recurrently, we can obtain \cref{ddim_determined_k}. In other words, we can obtain any point $x_t$ except $x_k$.
    \begin{equation}
        \boldsymbol{x}_{t}=\sqrt{\bar{\alpha}_{t}} \boldsymbol{x}_{0}+\sqrt{1-\bar{\alpha}_{t}} \cdot \frac{\boldsymbol{x}_{k}-\sqrt{\bar{\alpha}_{k}} \boldsymbol{x}_{0}}{\sqrt{1-\bar{\alpha}_{k}}}.
        \label{ddim_determined_k}
    \end{equation}
\end{proof}

We call the trajectory obtained from $\boldsymbol x_k$ \emph{groundtruth trajectory}.

Assuming that the point is $\boldsymbol{x}_{k}=\sqrt{\bar{a}_{k}} \boldsymbol{x}_{0}+\sqrt{1-\bar{a}_{k}} \overline{\boldsymbol{\epsilon}}_{k}$, to find a better groundtruth trajectory, we choose $k=0$ since the choice of $k$ is arbitrary, and approximate $\boldsymbol{\bar\epsilon}_0$ using \cref{eps_approx}.
\begin{equation}
    \boldsymbol{\epsilon}_{\boldsymbol{\theta}}\left(\sqrt{\bar{a}_{0}} \boldsymbol{x}_{0}+\sqrt{1-\bar{a}_{0}} \overline{\boldsymbol{\epsilon}}_{0}, 0\right) \approx \boldsymbol{\epsilon}_{\boldsymbol{\theta}}\left( \boldsymbol{x}_{0} , 0\right).
    \label{eps_approx}
\end{equation}

This choice is intuitive. First, $\bar\alpha_0$ is very close to $1$, making the approximation in \cref{eps_approx} valid. Second, the time $t=0$ is the closest timing to the original sample, so the model is likely to fit it better.

% Same as above, when $\sigma_t=0$, the denoising process becomes determined, and \cref{ddim_denoise} can be rewritten to \cref{ddim_denoise_determined}.

% Our approach also assumes that the samples in the training set have smaller loss like many other MIAs. In other words, the training samples fit the groundtruth trajectory better. We use $\ell_p\text{-norm}$ to measure the distance between any groundtruth point $\boldsymbol x_{t-t'}$ and the predicted point $\boldsymbol x'_{t-t'}$. It can be expressed by \cref{measure}. We use the notation $\boldsymbol x'_{t-t'}$ to indicate that $\boldsymbol x'_{t-t'}$ is predicted by the model from $\boldsymbol x_{t}$. We need to select a specific time $t-t'$ to apply this attack. We choose the closest $t'$ to $t$. In other word, $t'=t-1$. But later we will show that the choice of $t'$ is not important in discrete-time diffusion.

\subsection{Exposing Membership via Groundtruth Trajectory and Predicted Point}

Our approach assumes that the training set's samples have a smaller loss, similar to many other MIAs, meaning that the training samples align more closely with the groundtruth trajectory. We measure the distance between any groundtruth point $\boldsymbol x_{t-t'}$ and the predicted point $\boldsymbol x'_{t-t'}$ using the $\ell_p\text{-norm}$, which can be expressed by \cref{measure}. Here, $\boldsymbol x'_{t-t'}$ denotes the point predicted by the model from $\boldsymbol x_{t}$. To apply this attack, we need to select a specific time $t-t'$, and we choose the time $t'=t-1$ since it is the closest. However, we will demonstrate later that the choice of $t'$ is not significant in discrete-time diffusion.
\begin{equation}
    d_{t-t'} = \left\Vert \boldsymbol{x}_{t-t'} - \boldsymbol{x}'_{t-t'} \right\Vert_p.
    \label{measure}
\end{equation}

% We predict $\boldsymbol{x}'_{t-t'}$ from groundtruth point $\boldsymbol{x}_{t}$ simply by applying the deterministic version ($\sigma_t=0$) of DDIM denoising process \cref{ddim_denoise}.

To predict $\boldsymbol{x}'_{t-t'}$ from the groundtruth point $\boldsymbol{x}_{t}$, we apply the deterministic version ($\sigma_t=0$) of the DDIM denoising process \cref{ddim_denoise}.

We use method described in \cref{groundtruth_trajectory} to obtain the groundtruth point $\boldsymbol{x}_{t}$ and $\boldsymbol{x}_{t-t'}$. We then insert these points into \cref{measure}, giving us a simpler formula: 
$$
    \frac{\sqrt{1-\bar\alpha_{t-t'}}\sqrt{\bar\alpha_t}-\sqrt{1-\bar\alpha_t}\sqrt{\bar\alpha_{t-t'}}}{\sqrt{\bar\alpha_t}}\left\Vert \boldsymbol{\bar\epsilon}_0-\boldsymbol{\epsilon}_{\boldsymbol{\theta}}\left(\sqrt{\bar{a}_{t}} \boldsymbol{x}_{0}+\sqrt{1-\bar{a}_{t}} \overline{\boldsymbol{\epsilon}}_{0}, t\right) \right\Vert_p.
$$

If we ignore the coefficient, $t'$ disappears. Finally, the metric ignoring the coefficient reduces to \cref{final}, where samples with smaller $R_{t,p}$ are more likely to be training samples.
\begin{equation}
    R_{t,p} = \left\Vert  \boldsymbol{\epsilon}_{\boldsymbol{\theta}}\left( \boldsymbol{x}_{0} , 0\right)-\boldsymbol{\epsilon}_{\boldsymbol{\theta}}\left(\sqrt{\bar{a}_{t}} \boldsymbol{x}_{0}+\sqrt{1-\bar{a}_{t}}  \boldsymbol{\epsilon}_{\boldsymbol{\theta}}\left( \boldsymbol{x}_{0} , 0\right), t\right) \right\Vert_p.
    \label{final}
\end{equation}

% \begin{equation}
    
%     \label{x_t_groundtruth}
% \end{equation}

% Now, we need to find $\boldsymbol {\bar\epsilon}_k$. Since the choice of $k$ is free, we choose $k=0$ and approximate $\boldsymbol {\bar\epsilon}_0$ by \cref{eps_approx}.
% We can also get a more precisely $\boldsymbol {\bar\epsilon}_0$ by optimize \cref{optim} with $t=0$. Since the form of \cref{final} can be considered as loss function, if not otherwise specified, 

% The choice of $p$ refers to the choice of loss function during training.

Since $\epsilon$ is initialized in time $t=0$, we call our method Proximal Initialization Attack (PIA).

% \textbf{Normalization} ~ $\boldsymbol{\epsilon}_{\boldsymbol{\theta}}\left( \boldsymbol{x}_{0} , 0\right)$ may not satisfy the standard normal distribution, so we use \cref{PIAN} to normalize it. $N$ is the number of the elements of sample. For example, for a image, it is $h \times w$. We call this method PIAN (PIA Normalized). Strictly speaking, this normalization 
% cannot gaurantee $\hat{\boldsymbol{\epsilon}}_{\boldsymbol{\theta}}\left( \boldsymbol{x}_{0} , 0\right)\sim \mathcal{N}(\mathbf{0}, \mathbf{I})$. However, since the every element of $\bar\epsilon_t$ in training loss \cref{loss_func} satisfies independent and identically distributed, we think this normalization is reasonable.

\textbf{Normalization} The values of $\boldsymbol{\epsilon}_{\boldsymbol{\theta}}\left(\boldsymbol{x}_{0}, 0\right)$ may not conform to a standard normal distribution, so we use \cref{PIAN} to normalize them. $N$ represents the number of elements in the sample, such as $h \times w$ for an image. We refer to this method as PIAN (PIA Normalized). Although this normalization cannot guarantee that $\hat{\boldsymbol{\epsilon}}_{\boldsymbol{\theta}}\left(\boldsymbol{x}_{0},0\right)\sim\mathcal{N}(\mathbf{0},\mathbf{I})$, we deem it reasonable since each element of $\bar\epsilon_t$ in the training loss \cref{loss_func} is identically and independently distributed.
\begin{equation}
    \hat{\boldsymbol{\epsilon}}_{\boldsymbol{\theta}}\left( \boldsymbol{x}_{0} , 0\right)=
    \frac{\boldsymbol{\epsilon}_{\boldsymbol{\theta}}( \boldsymbol{x}_{0} , 0)}{\mathbb E_{ x  \sim \mathcal N( 0, 1)}(| x|)\frac{\Vert\boldsymbol{\epsilon}_{\boldsymbol{\theta}}( \boldsymbol{x}_{0} , 0)\Vert_1}{N}} = N \sqrt{\frac{\pi}{2}} \frac{\boldsymbol{\epsilon}_{\boldsymbol{\theta}}( \boldsymbol{x}_{0} , 0)}{\Vert\boldsymbol{\epsilon}_{\boldsymbol{\theta}}( \boldsymbol{x}_{0} , 0)\Vert_1}.
    \label{PIAN}
\end{equation}

To apply our attack, we first evaluate the value of $R_{t,p}$ on a sample, and use an indicator function:
\begin{equation}
    f(\boldsymbol{x}, m)=\mathds{1}[R_{t,p} < \tau]. \label{threshold}
\end{equation}
This indicator means we consider whether a sample is in the training set if $R_{t,p}$ is smaller than a threshold $\tau$. $R_{t,p}$ is obtained from ${\boldsymbol{\epsilon}}_{\boldsymbol{\theta}}\left( \boldsymbol{x}_{0} , 0\right)$ (PIA) or $\hat{\boldsymbol{\epsilon}}_{\boldsymbol{\theta}}\left( \boldsymbol{x}_{0} , 0\right)$ (PIAN).

% \textbf{Error Based Attack (remove this section?)} ~ In fact, Naive Atack, SecMI and our method can be reduced to the same form.

% \begin{equation}
%     E_t=\Vert \boldsymbol\epsilon - \boldsymbol\epsilon_\theta(f(\boldsymbol x_0, \boldsymbol\epsilon, \theta, t), t) \Vert_p
%     \label{reduced}
% \end{equation}

% In \cref{reduced}, $\boldsymbol{\epsilon}$ is a initialized noise related $\boldsymbol x_0$, $t$ and the model $\theta$. $f$ is a function to estimate the point $x_t$ from $x_0$. This suggests that for diffusion models, choosing better $\epsilon$ and $f$ may lead to better results.

\subsection{For Continuous-Time Diffusion Model}

Recently, some diffusion models are trained with continuous time. As demonstrated in \cite{Score_Based}, the diffusion process with continuous time can be defined by a stochastic differential equation (SDE) as $d\boldsymbol x_t= \boldsymbol f_t(\boldsymbol{x}_t)dt + g_t d\boldsymbol w_t$, where $\boldsymbol w_t$ is a Brownian process. One of the reverse process is $
    d \boldsymbol{x}_t=\left(\boldsymbol{f}_{t}(\boldsymbol{x}_t)-\frac{1}{2}\left(g_{t}^{2}+\sigma_{t}^{2}\right) \nabla_{\boldsymbol{x}_t} \log p_{t}(\boldsymbol{x}_t)\right) d t+\sigma_{t} d \boldsymbol{w}$. When $\sigma_t=0$, this formula becomes an ordinary differential equation (ODE): $
    d \boldsymbol{x}_t=\left(\boldsymbol{f}_{t}(\boldsymbol{x}_t)-\frac{1}{2}g_{t}^{2} \nabla_{\boldsymbol{x}_t} \log p_{t}(\boldsymbol{x}_t)\right) d t
    \label{ODE}$.
Continuous-time diffusion model train an $\boldsymbol{s}_{\boldsymbol{\theta}}$ to approximate $\nabla_{\boldsymbol{x}_t} \log p_{t}(\boldsymbol{x}_t)$, so the loss function will be:
\begin{align*}
    L = \mathbb{E}_{\boldsymbol{x}_{0}, \boldsymbol{x}_{t} \sim p\left(\boldsymbol{x}_{t} \mid \boldsymbol{x}_{0}\right) \bar{p}\left(\boldsymbol{x}_{0}\right)}\left[\left\|\boldsymbol{s}_{\boldsymbol{\theta}}\left(\boldsymbol{x}_{t}, t\right)-\nabla_{\boldsymbol{x}_{t}} \log p\left(\boldsymbol{x}_{t} \mid \boldsymbol{x}_{0}\right)\right\|^{2}\right].
    % \label{loss sde}
    \vspace{-1mm}
\end{align*}

Replacing $\nabla_{\boldsymbol{x}_t} \log p_{t}(\boldsymbol{x}_t)$ with $\boldsymbol{s}_{\boldsymbol{\theta}}\left(\boldsymbol{x}_{t}, t\right)$, the inference procedure become the following equation:
\begin{equation}
    d \boldsymbol{x}_t=\left(\boldsymbol{f}_{t}(\boldsymbol{x}_t)-\frac{1}{2}g_{t}^{2} \boldsymbol{s}_{\boldsymbol{\theta}}(\boldsymbol{x}_{t}, t)\right) d t.
    \label{ODE_inference}
\end{equation}

% Usually, distribution $p(\boldsymbol{x}_{t} | \boldsymbol{x}_{0})$ is set to be the same as in DDPM. So, the loss of continuous-time diffusion model and the loss of concrete-diffusion model \cref{loss_func} become similar. Since DDPM and the diffusion model described by SDE share the similar loss, our method can apply to continuous-time diffusion model. Because of the different diffusion process, $D_t$ has some differences compared to \cref{simple}. From \cref{ODE_inference}, we can get the following equation:
% $
%     \boldsymbol x_{t-t'} - \boldsymbol x_{t} \approx d\boldsymbol x_t = \left(\boldsymbol{f}_{t}(\boldsymbol{x}_t)-\frac{1}{2}g_{t}^{2} \boldsymbol{s}_{\boldsymbol{\theta}}\left(\boldsymbol{x}_{t}, t\right)\right) dt.
% $
% Put this equation into \cref{measure}, we can get following equation:

The distribution $p(\boldsymbol{x}_{t} | \boldsymbol{x}_{0})$ is typically set to be the same as in DDPM for continuous-time diffusion models. Therefore, the loss of the continuous-time diffusion model and the loss of the concrete-diffusion model \cref{loss_func} are similar. Since DDPM and the diffusion model described by SDE share a similar loss, our method can be applied to continuous-time diffusion models. However, due to the different diffusion process, $R_{t,p}$ differs from \cref{final}. From \cref{ODE_inference}, we obtain the following equation:
$
    \boldsymbol x_{t-t'} - \boldsymbol x_{t} \approx d\boldsymbol x_t = \left(\boldsymbol{f}_{t}(\boldsymbol{x}_t)-\frac{1}{2}g_{t}^{2} \boldsymbol{s}_{\boldsymbol{\theta}}\left(\boldsymbol{x}_{t}, t\right)\right) dt.
$
By substituting this equation into \cref{measure}, we obtain the following equation:
\begin{equation}
    \Vert \boldsymbol x_{t-t'} - \boldsymbol x'_{t-t'} \Vert_p \approx \left\Vert \left(\boldsymbol{f}_{t}(\boldsymbol{x}_t)-\frac{1}{2}g_{t}^{2} \boldsymbol{s}_{\boldsymbol{\theta}}\left(\boldsymbol{x}_{t}, t\right)\right) d t + \boldsymbol x_{t} - \boldsymbol x'_{t-t'} \right\Vert_p.
    \label{approx}
\end{equation}

% Similar to discrete-time diffusion, we choose the closest $t'$ to $t$.
By ignoring the high-order infinitesimal term $\boldsymbol x_{t} - \boldsymbol x'_{t-t'}$ in \cref{approx}, we can obtain $\Vert \boldsymbol x_{t-t'} - \boldsymbol x'_{t-t'} \Vert_p \approx \left\Vert \left(\boldsymbol{f}_{t}(\boldsymbol{x}_t)-\frac{1}{2}g_{t}^{2} \boldsymbol{s}_{\boldsymbol{\theta}}\left(\boldsymbol{x}_{t}, t\right)\right) d t\right\Vert_p dt$. We ignore $dt$ and use the following attack metric:
\begin{equation}
    R_{t,p} = \left\Vert\boldsymbol{f}_{t}(\boldsymbol{x}_{t})-\frac{1}{2}g_{t}^{2} \boldsymbol{s}_{\boldsymbol{\theta}}\left(\boldsymbol{x}_{t}, t\right)\right\Vert_p,
    \label{sde_final}
\end{equation}

where $\boldsymbol{x}_{t}$ is obtained from the output of $\boldsymbol{s}_{\boldsymbol{\theta}}(\boldsymbol{x}_{0}, 0)$, similar to the discrete-time diffusion case.

\section{Experiment}

In this section, we evaluate the performance of PIA and PIAN and robustness of TTS models across various datasets and settings. The detailed experimental settings, including datasets, models, and hyper-parameter settings can be found in Appendix A.

\subsection{Evaluation Metrics}

We follow the most convincing metrics used in MIAs~\cite{carlini2023extracting}, including AUC, the True Positive Rate (TPR) when the False Positive Rate (FPR) is 1\%, i.e., TPR @ 1\% FPR, and TPR @ 0.1\% FPR.

% \begin{figure}[h]
%     \centering
%     % \setlength{\subfigbottomskip}{6pt}
%     % \setlength{\subfigcapskip}{4pt}
%     % \setlength{\tabcolsep}{0pt}
%     % \def\arraystretch{0}
%     % \setlength{\subfigbottomskip}{12pt}
%     \subfigure[]{\includegraphics[width=0.49\linewidth]{images/gradtts.pdf}}
%     \subfigure[]{\includegraphics[width=0.49\linewidth]{images/image.pdf}}
    
%     \vspace*{-3mm}
%     \caption{gradtts} % 图片标题
%     \label{fig:example} % 图片标签，方便引用
% \end{figure}

\subsection{Proximal Initialization Attack Performance}
\label{performace}
% Please add the following required packages to your document preamble:
% \usepackage[normalem]{ulem}
% \useunder{\uline}{\ul}{}
\begin{table}[t!]
\centering
\caption{Performance of different methods on Grad-TTS. TPR@x\% is the abbreviation for TPR@x\% FPR.}
\label{gradtts_performence}
\adjustbox{width=0.95\textwidth}{
\begin{tabular}{c|cc|cc|cc|c}
\toprule
            & \multicolumn{2}{c|}{LJspeech} & \multicolumn{2}{c|}{VCTK}     & \multicolumn{2}{c|}{LibriTTS} &         \\ \midrule
Method      & AUC           & TPR@1\% FPR   & AUC           & TPR@1\% FPR   & AUC           & TPR@1\% FPR  & Query    \\ \midrule
NA~\cite{matsumoto2023membership}     & 99.4          & 93.6          & 83.4          & 6.1           & 90.2          & 9.1      & \textbf{1}     \\
SecMI~\cite{SecMI}      & {\ul 99.5}          & 94.0          & 87.0          & 14.8          & {\ul 93.9}          & 19.7    & 60+2        \\
PIA    & \textbf{99.6} & {\ul 94.2}    & {\ul 87.8} & \textbf{20.6} & \textbf{95.4}    & {\ul 30.0}   & {\ul 1+1}  \\
PIAN   & { 99.3}    & \textbf{95.7} & \textbf{88.1}    & {\ul 19.6}    & { 93.4} & \textbf{44.7} & {\ul 1+1} \\ \bottomrule
\end{tabular}

}

% \vspace{-6mm}
\end{table}

\begin{table}[t!]
\centering
\caption{Performance of the different methods on DDPM.}
\label{image_generation}
\adjustbox{width=0.8\textwidth}{

% \vspace{-10mm}

\begin{tabular}{c|cc|cc|cc|c}
\toprule
\multicolumn{1}{l|}{} & \multicolumn{2}{c|}{CIFAR10}  & \multicolumn{2}{c|}{TN-IN}    & \multicolumn{2}{c|}{CIFAR100} &            \\ \midrule
Method                & AUC           & TPR@1\% FPR    & AUC           & TPR@1\% FPR    & AUC           & TPR@1\% FPR    & Query      \\ \midrule
NA                    & 84.7          & 6.85          & 84.9          & 10.0          & 82.3          & 9.6           & \textbf{1} \\ 
SecMI                 & {\ul 88.1}    & 9.11          & {\ul 89.4}    & 12.7          & {\ul 87.6}    & 11.1          & 10+2       \\
PIA                   & \textbf{88.5} & {\ul 13.7}    & \textbf{89.6} & {\ul 17.1}    & \textbf{89.4} & {\ul 19.6}    & {\ul 1+1}  \\
PIAN                  & 87.8          & \textbf{31.2} & 88.2          & \textbf{32.8} & 86.5          & \textbf{22.2} & {\ul 1+1}  \\ \bottomrule
\end{tabular}

}
\end{table}

\begin{table}[]
\centering
\caption{Performance of different methods on stable diffusion.}
\label{stable_diffusion}
\adjustbox{width=0.8\textwidth}{
\begin{tabular}{c|cc|cc|cc|c}
\toprule
         & \multicolumn{2}{c|}{Laion5}         & \multicolumn{2}{c|}{Laion5 w/o text} & \multicolumn{2}{c|}{Laion5 Blip text}    &            \\ \midrule
Method    & AUC                 & TPR@1\% FPR    & AUC            & TPR@1\% FPR          & AUC                 & TPR@1\% FPR     & Query        \\ \midrule
NA      & 66.3                & 14.8          & 65.2           & 13.3                & 68.2                & 16.2        & \textbf{1}        \\
SecMI     & {\ul {69.1}} & {\ul 16.1}    & {\ul 71.6}     & {\ul {14.5}} & {\ul {71.6}} & {\ul {17.8}} & 10+2     \\
PIA      & \textbf{70.5}       & \textbf{18.1} & \textbf{73.9}  & \textbf{19.8}       & \textbf{73.3}       & \textbf{20.2}    & {\ul 1+1}   \\
PIAN     & 56.7                & 4.8           & 58.8           & 3.2                 & 55.3                & 3.2        & {\ul 1+1}         \\ \bottomrule
\end{tabular}
}
\end{table}

We train TTS models on the LJSpeech, VCTK, and LibriTTS datasets. We summarize the AUC and TPR @ 1\% FPR results on GradTTS, a continuous-time diffusion model, in \cref{gradtts_performence}. We employ NA to denote Naive Attack. Compared to SecMI, PIA and PIAN achieve slightly better AUC performance, and significantly higher TPR @ 1\% FPR performance, i.e., 5.4\% higher for PIA and 10.5\% higher for PIAN on average. However, our proposed method only requires $1+1$ queries, just one more query than Naive Attack, and has a computational consumption of only 3.2\% of SecMI. Both methods outperform SecMI and Naive Attack.

For DDPM, a discrete-time diffusion model, we present the results in \cref{image_generation}. For this model, PIA performs slightly better than SecMI in terms of AUC but has a distinctly higher TPR @ 1\% FPR than SecMI, i.e. 5.8\% higher on average than SecMI. For PIAN, the AUC performance is slightly lower than PIA, but higher than SecMI, and the TPR @ 1\% FPR performance is significantly better than SecMI, i.e. 17.8\% higher on average than SecMI. Similar to the previous case, our attack only requires two queries on DDPM and the computational consumption is 17\% of SecMI. Both methods outperform SecMI and Naive Attack.

For stable diffusion, we present the results in \cref{stable_diffusion}. We evaluated stable diffusion on Laion5 (training dataset) and COCO (evaluation dataset). Details are put into A.2. We tested three scenarios: knowing the ground truth text (Laion5), not knowing the ground truth text (Laion5 w/o text), and generating text through blip (Laion5 Blip text). PIA achieved the best results. PIA performs slightly better than SecMI in terms of AUC, i.e. 1.8\% higher on average, but has a distinctly higher TPR @ 1\% FPR than SecMI, i.e. 3.2\% higher on average. Besides, our attack only requires two queries on DDPM and the computational consumption is 17\% of SecMI.

However, PIAN does not work well in stable diffusion. PIAN based on the fact that we added noise that follows a normal distribution during training, and we use \cref{PIAN} to rescale the $\epsilon$ to normal distribution. However, rescaling is a rough operation and may not always transform into a normal distribution. Thus, some other transforms might have better performance. Additionally, the model's output might be more accurate before the rescaling.

We highly recommend using PIA as the preferred method for conducting attacks, because it is directly derived. It will always yield the desired results. But PIAN can be another choice, since it has better performance at TPR @ 1\% FPR metric than PIA on some models.

\vspace{-2mm}

\subsection{Ablation Study}

\begin{figure}
    \begin{subfigure}{0.32\textwidth}
        \centering 
        \includegraphics[width=\textwidth]{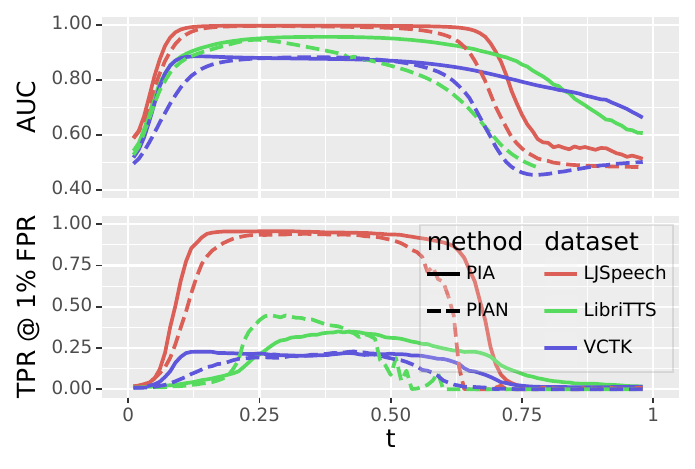} % 1st subfigure: \includegraphics{fig}...
    %   \vspace{-6mm}
    \captionsetup{skip=0pt, width=.95\linewidth}
      \subcaption[]{The results of PIA and PIAN on Grad-TTS for different values of $t$ and different datasets.}
        % \includegraphics[width=0.49\linewidth]{images/gradtts.pdf}
        % \subcaption[subfigcapskip=500pt]{A blue square.}
        \label{ablation_a}
    \end{subfigure}
    % \par\bigskip % maximise vertical space here instead
    \begin{subfigure}{0.32\textwidth}
    \captionsetup{skip=0pt, width=.95\linewidth}
        \centering \includegraphics[width=\textwidth]{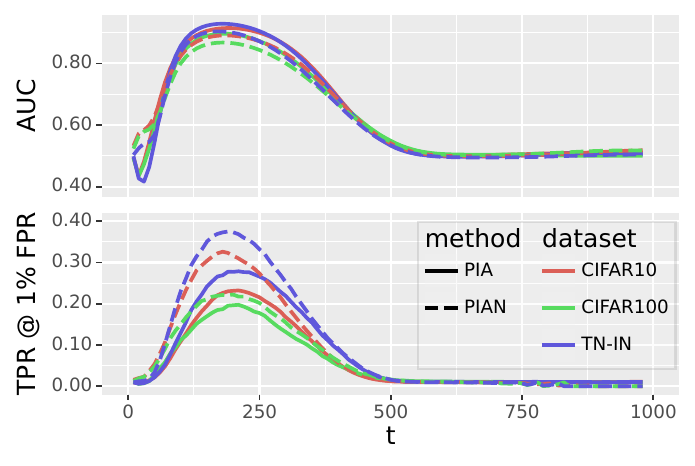}% 2nd subfigure: \includegraphics{fig}...
        % \vspace*{-3mm}
        % \vspace{-6mm}
        \subcaption[]{The results of PIA and PIAN on DDPM for different values of t and different datasets.}
        \label{ablation_b}
    \end{subfigure}
    \begin{subfigure}{0.32\textwidth}
    \captionsetup{skip=0pt, width=.95\linewidth}
        \centering \includegraphics[width=\textwidth]{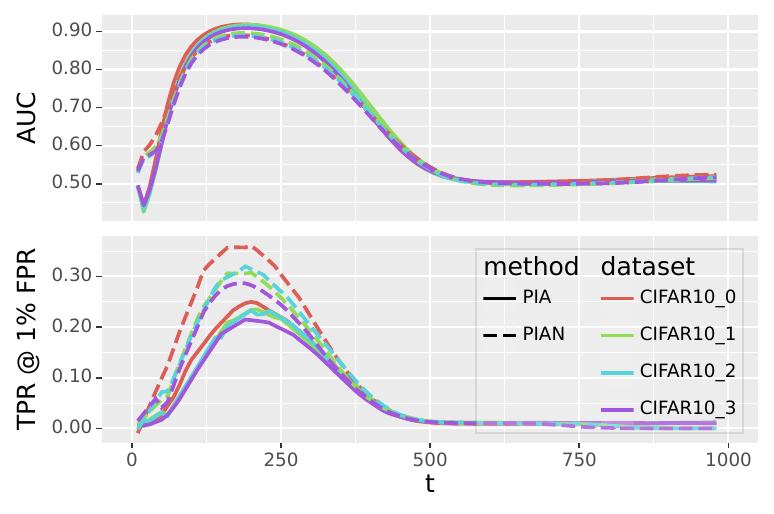}% 2nd subfigure: \includegraphics{fig}...
        \subcaption[]{The results of PIA and PIAN on DDPM for different values of $t$ and different CIFAR10 splits.}
        \label{ablation_c}
    \end{subfigure}
    
\caption{The performance of PIA and PIAN as $t$ varies. The top row shows the results for AUC, and the bottom row shows the results for TPR @ 1\% FPR.}
        \label{ablation}
    \vspace{-5mm}
\end{figure}

% \begin{figure}
%     \begin{subfigure}{0.49\textwidth}
%         \centering 
%         \includegraphics[width=\textwidth]{images/gradtts_libritts_different.pdf} % 1st subfigure: \includegraphics{fig}...
%     %   \vspace{-6mm}
%     \captionsetup{skip=0pt, width=.95\linewidth}
%       \subcaption[]{The results of PIA and PIAN on Grad-TTS for different values of $t$ and different LibriTTS splits.}
%     \end{subfigure}
%     % \par\bigskip % maximise vertical space here instead
%     \begin{subfigure}{0.49\textwidth}
%     \captionsetup{skip=0pt, width=.95\linewidth}
%         \centering \includegraphics[width=\textwidth]{images/image_CIFAR10_different.pdf}% 2nd subfigure: \includegraphics{fig}...
%         \subcaption[]{The results of PIA and PIAN on DDPM for different values of $t$ and different CIFAR10 splits.}
%     \end{subfigure}
% \captionsetup{skip=5pt}
% \caption{The performance of our method on 4 random split LibriTTS/CIFAR10 dataset as $t$ varies. The top row shows the results for AUC, and the bottom row displays the results for TPR @ 1\% FPR.}
%         \label{ablation_different_split}
%         \vspace{-3mm}
% \end{figure}

\begin{figure}
\centering
    \begin{subfigure}{0.45\textwidth}
        \centering 
        \includegraphics[width=\textwidth]{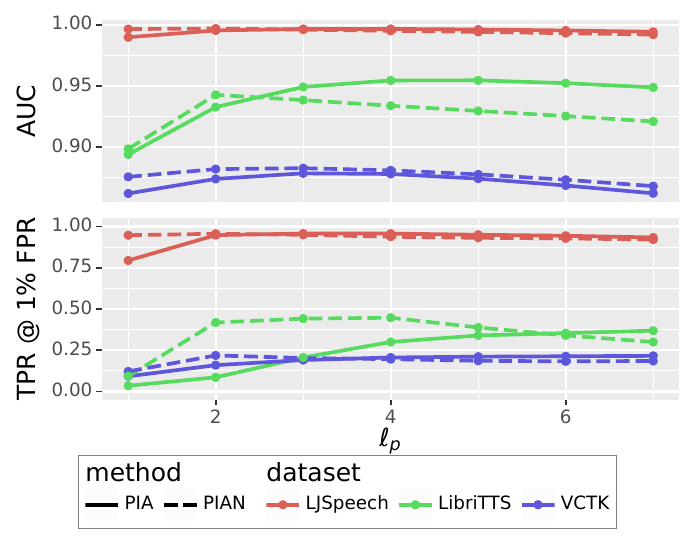} % 1st subfigure: \includegraphics{fig}...
    %   \vspace{-6mm}
    \captionsetup{skip=0pt, width=.95\linewidth}
      \subcaption[]{The results of PIA and PIAN on Grad-TTS for different values of $\ell_p$-norm.}
    \end{subfigure}
    % \par\bigskip % maximise vertical space here instead
    \begin{subfigure}{0.45\textwidth}
    \hspace{3mm}
    \captionsetup{skip=0pt, width=.95\linewidth}
        \centering \includegraphics[width=\textwidth]{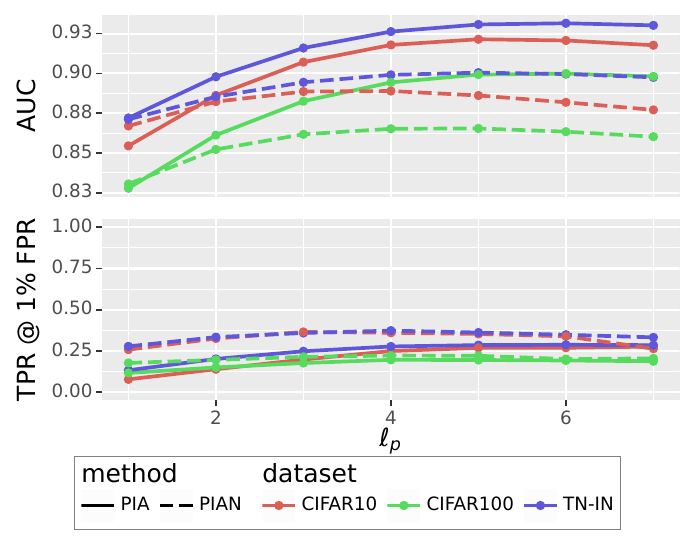}% 2nd subfigure: \includegraphics{fig}...
        \subcaption[]{The results of PIA and PIAN on DDPM for different values of $\ell_p$-norm.}
    \end{subfigure}
\captionsetup{skip=5pt}
\caption{The performance of our method as $\ell_p$-norm varies. The top row shows the results for AUC, and the bottom row displays the results for TPR @ 1\% FPR.}
        \label{ablation_lp_norm}
        \vspace{-2mm}
\end{figure}

Our proposed method has three hyper-parameters: $t$ and the $\ell_p$-norm used in the attack metrics $R_{t,p}$ presented in \cref{final,sde_final}. The threshold $\tau$ presented in \cref{threshold}.

\textbf{Impact of $t$} ~ To evaluate the impact of $t$, we attack the target model at intervals of $0.01\times T$ from $0$ to $T$ and report the results across different models and datasets. We demonstrate the performance of our proposed method on two different models: GradTTS, a continuous-time diffusion model used for audio, in \cref{ablation_a}; and DDPM, a discrete-time diffusion model employed for images, in \cref{ablation_b}. The results indicate that our method produces a consistent pattern in the same model across different datasets, whether PIA or PIAN. Specifically, for GradTTS, both AUC and TPR @ 1\% FPR exhibit a rapid increase at the beginning as $t$ increases, followed by a decline around $t=0.5$. For DDPM, AUC and TPR @ 1\% FPR also demonstrate a rapid increase at the beginning as $t$ increases, followed by a decline around $t=200$. In \cref{ablation_c}, we randomly partition the CIFAR10 dataset four times and compare the performance of each partition. Consistent with the previous results, our method exhibits a similar trend across the different splits.

\begin{table}[]
\centering
\caption{The variation of Attack Success Rate (ASR) and TPR/FPR on the victim model with the threshold determined by the surrogate model.}
\label{shadow_model}
\adjustbox{width=0.9\textwidth}{
\begin{tabular}{c|cccc|cccc}
\toprule
% \multirow{3}{*}{\diagbox[height=3\line,width=\widthof{Shadow modelaaa}]{}{}}
& \multicolumn{4}{c|}{PIA}                                            & \multicolumn{4}{c}{PIAN}                                           \\ \cmidrule{2-9} 
                  & \multicolumn{2}{c|}{LibriTTS}        & \multicolumn{2}{c|}{CIFAR10} & \multicolumn{2}{c|}{LibriTTS}        & \multicolumn{2}{c}{CIFAR10} \\ \cmidrule {2-9} 
                  & ASR  & \multicolumn{1}{c|}{TPR/FPR}  & ASR         & TPR/FPR        & ASR  & \multicolumn{1}{c|}{TPR/FPR}  & ASR         & TPR/FPR       \\ \midrule
Surrogate model      & 89.5 & \multicolumn{1}{c|}{32.2/1}   & 78.5        & 16.5/1         & 88.3 & \multicolumn{1}{c|}{26.2/1}   & 76.9        & 19.0/1        \\
Victim model      & 89.1 & \multicolumn{1}{c|}{32.6/1.1} & 78.3        & 16.8/1.1       & 88.2 & \multicolumn{1}{c|}{24.5/0.9} & 76.8        & 19.0/1        \\ \bottomrule
\end{tabular}
}
% \vspace{-4mm}
\end{table}

\textbf{Impact of $\ell_p$-norm} ~ In \cref{ablation_lp_norm}, we compare the results obtained on $\ell_p$-norm using the $p=1$ to $7$, with the choice of $t$ being the same as in \cref{performace}. The results indicate an increase in performance at $\ell_1$-norm, followed by a decline after the $p=5$. It reveals that the combined effect of both large and small differences exhibits a synergistic influence when present in an appropriate ratio.

\textbf{Determining the value of $\tau$} ~ In \cref{shadow_model}, we present the variation of Attack Success Rate (ASR) and TPR/FPR on the victim model with the $\tau$ determined by the surrogate model. Specifically, we will randomly split the corresponding dataset into two halves four times, resulting in four different train-test splits. We will train four models using these splits. One of the models will be selected as the surrogate model, from which we will obtain the threshold. We will then use this $\tau$ to attack the other three victim models and record the average values. The results indicate that our method achieves promising results when using the $\tau$ selected from the surrogate model.

\begin{table}[tbp]
% \vspace{-6mm}
\centering
\caption{Comparison of different models. AUC is the result on the LJSpeech/TinyImageNet dataset.}
\label{compare}
\adjustbox{width=0.8\textwidth}{
\begin{tabular}{cccccc}
\toprule
Model                       & Size   & T       & Output          & Segmentation Length & Best AUC    \\ \midrule
DDPM                        & 35.9M  & 1000    & Image           & N/A            & 92.6        \\
\multicolumn{1}{l}{GradTTS} & 56.7M  & $[0,1]$ & Mel-spectrogram & 2s             & 99.6        \\
DiffWave                    & 30.3M  & 50      & Audio           & 0.25s          & 52.4        \\
FastDiff                    & 175.4M & 1000    & Audio           & 1.2s           & 54.4 \\ \bottomrule
\end{tabular}
}
\label{model_comparision}
% \vspace{-2mm}
\end{table}

% \vspace{-1mm}
\subsection{Which Type of Model Output is More Robust?}
% \vspace{-1mm}

There are generally two forms of output in TTS: mel-spectrograms and audio. In \cref{model_comparision}, we summarize the model details and best results of our proposed method on three TTS models using the LJSpeech dataset and the DDPM model on the TinyImageNet dataset. We only report the results of our method since it achieves better performance most of the time.

As shown in \cref{model_comparision}, with the same training and hold-out data, GradTTS achieves an AUC close to 100, while DiffWave and FastDiff only achieve the performance slightly above 50, which is close to random guessing. 
However, DiffWave has a similar size to DDPM and GradTTS, and FastDiff has similar $T$ with DDPM. Additionally, FastDiff has similar segmentation length to GradTTS. Thus, we believe that these hyperparameters are not the decisive parameters for the model's robustness. It is obvious that the output of GradTTS and DDPM is image-like. \cref{mel_example} provides an example of mel-spectrogram. The deep reasons why these models exhibit robustness can be further explored. We report these results hoping that they may inspire the design of models with MIA robustness.

\begin{figure}

        \centering 
        \includegraphics[width=0.5\textwidth]{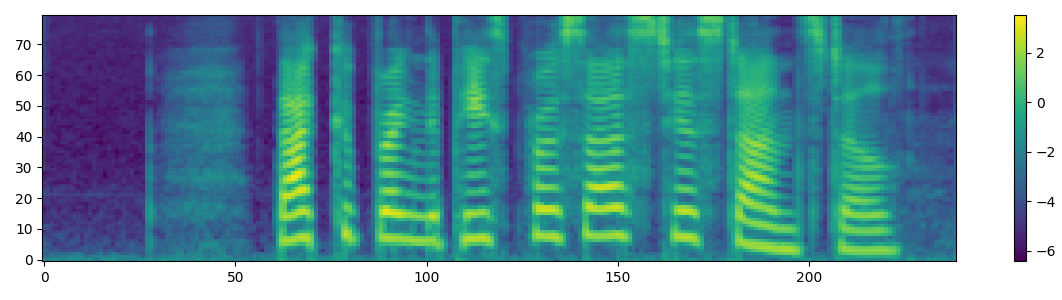} % 1st subfigure: \includegraphics{fig}...
        \vspace{-3mm}
\caption{An example of mel-spectrogram.}
\vspace{-4mm}
    \label{mel_example}
\end{figure}

\begin{figure}
\centering
    \begin{subfigure}{0.45\textwidth}
        \centering 
        \includegraphics[width=\textwidth]{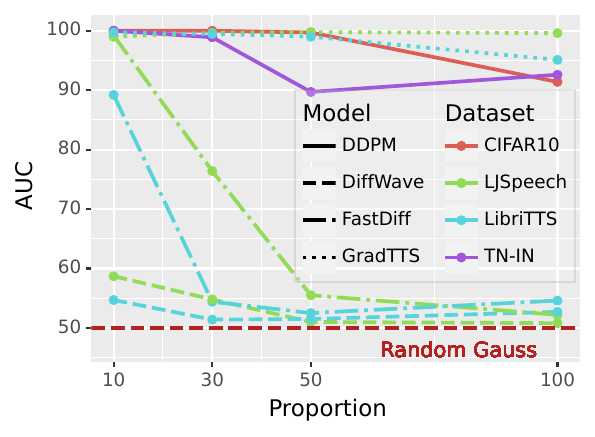} % 1st subfigure: \includegraphics{fig}...
    %   \vspace{-6mm}
    \captionsetup{skip=0pt, width=.95\linewidth}
      \subcaption[]{The results of PIA and PIAN on Grad-TTS for different training and evaluation sample numbers.}
        % \includegraphics[width=0.49\linewidth]{images/gradtts.pdf}
        % \subcaption[subfigcapskip=500pt]{A blue square.}
        % \label{ablation_a}
    \end{subfigure}
    \hspace{3mm}
    % \par\bigskip % maximise vertical space here instead
    \begin{subfigure}{0.45\textwidth}
    \captionsetup{skip=0pt, width=.95\linewidth}
        \centering \includegraphics[width=\textwidth]{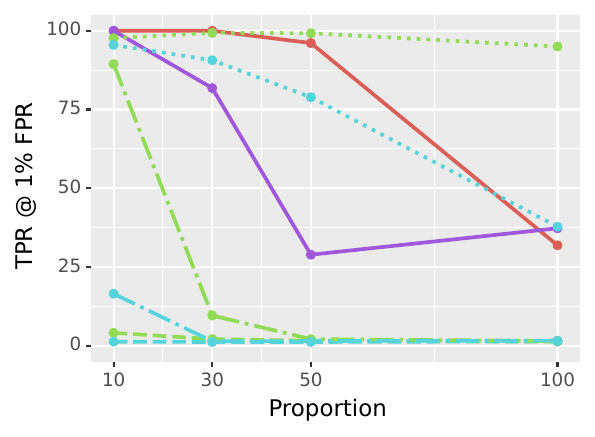}% 2nd subfigure: \includegraphics{fig}...
        % \vspace*{-3mm}
        % \vspace{-6mm}
        \subcaption[]{The results of PIA and PIAN on DDPM for different training and evaluation sample numbers.}
        % \label{ablation_b}
    \end{subfigure}
    
\caption{The performance of our method for different training and evaluation sample numbers. The top row shows the results for AUC, and the bottom row displays the results for TPR @ 1\% FPR.}
        \label{ablation_different}
\vspace{-4mm}
\end{figure}

We also explore the attack performance with various training and evaluation sample numbers. We select 10\%, 30\%, 50\%, and 100\% of the samples from the complete dataset. In each split, half of all samples are used for training, and the other half are utilized as a hold-out set. The results are presented in \cref{ablation_different}. As we can see, when only 10\% of the data is used, relatively high AUC and TPR @ 1\% FPR can be achieved. Additionally, we find that the AUC and TPR @ 1\% FPR decrease as the proportion of selected samples in the total dataset increases. However, for GradTTS and DDPM, the decrease is relatively gentle, while for DiffWave and FastDiff, the decrease is rapid. In other words, the robustness increases rapidly with the increase of training samples.

\vspace{-2mm}

\section{Conclusion}

In this paper, we propose an efficient membership inference attack method for diffusion models, namely Proximal Initialization Attack (PIA) and its normalized version, PIAN. We demonstrate its effectiveness on a continuous-time diffusion model, GradTTS, and two discrete-time diffusion models, DDPM and Stable Diffusion. Experimental results indicate that our proposed method can achieve similar AUC performance to SecMI and significantly higher TPR @ 1\% FPR with the cost of only 2 queries, which is much faster than the 12\textasciitilde62 queries required for SecMI in this paper. Additionally, we analyze the vulnerability of models in TTS, an audio generation tasks. The results suggest that diffusion models with the image-like output (mel-spectrogram) are more vulnerable than those with the audio output. Therefore, for privacy concerns, we recommend employing models with audio outputs in text-to-speech tasks.

\textbf{Limitation and Broader Impacts} ~ 
The purpose of our method is to identify whether a given sample is part of the training set. This capability can be leveraged to safeguard privacy rights by detecting instances of personal information being unlawfully used for training purposes. However, it is important to note that our method could also potentially result in privacy leaking. For instance, this could occur when anonymous data is labeled by determining whether a sample is part of the training set or as a part of data reconstruction attack. It is worth mentioning that our method solely relies on the diffusion model's output as we discussed in the threat model, but it does require the intermediate output. This dependency on the intermediate output may pose a limitation to our method.

\bibliographystyle{plain}

% \section*{References}

% References follow the acknowledgments in the camera-ready paper. Use unnumbered first-level heading for
% the references. Any choice of citation style is acceptable as long as you are
% consistent. It is permissible to reduce the font size to \verb+small+ (9 point)
% when listing the references.
% Note that the Reference section does not count towards the page limit.
% \medskip

% {
% \small

% [1] Alexander, J.A.\ \& Mozer, M.C.\ (1995) Template-based algorithms for
% connectionist rule extraction. In G.\ Tesauro, D.S.\ Touretzky and T.K.\ Leen
% (eds.), {\it Advances in Neural Information Processing Systems 7},
% pp.\ 609--616. Cambridge, MA: MIT Press.

% [2] Bower, J.M.\ \& Beeman, D.\ (1995) {\it The Book of GENESIS: Exploring
%   Realistic Neural Models with the GEneral NEural SImulation System.}  New York:
% TELOS/Springer--Verlag.

% [3] Hasselmo, M.E., Schnell, E.\ \& Barkai, E.\ (1995) Dynamics of learning and
% recall at excitatory recurrent synapses and cholinergic modulation in rat
% hippocampal region CA3. {\it Journal of Neuroscience} {\bf 15}(7):5249-5262.
% }

% %%%%%%%%%%%%%%%%%%%%%%%%%%%%%%%%%%%%%%%%%%%%%%%%%%%%%%%%%%%%

\newpage
\appendix
\section*{Appendix}

\section{Datasets and Diffusion Models}

\setcounter{footnote}{0} 

For TTS, we evaluate three commonly used datasets: LJSpeech, VCTK, and a subset of LibriTTS called libritts-lean-100. We test three models: GradTTS \footnote{\url{https://github.com/huawei-noah/Speech-Backbones/tree/main/Grad-TTS}}, FastDiff \footnote{\url{https://github.com/Rongjiehuang/FastDiff}}, and DiffWave \footnote{\url{https://github.com/lmnt-com/diffwave}}. For image generation, we evaluate the CIFAR10, CIFAR100 and TinyImageNet datasets using the same DDPM model as \cite{SecMI}, and Laion5, COCO for stable diffusion \cite{stable_diffusion}. Unless otherwise specified, we randomly select half of the samples as a training set and the other half as the hold-out set.

\subsection{Implementations Details} 

For the audio generation models, we use their codes from the official repositories and apply the default hyperparameters for all models except for the hyperparameters we mentioned. The training iterations were set to 1,000,000, due to the default value for the three audio generative models are all around this. For DDPM, all settings are the same as those in \cite{SecMI}.

On GradTTS and DDPM, we utilized a consistent attack time $t$ across different datasets for the same model. On DDPM, for Naive Attack, we set $t=200$. For SecMI, we set $t=100$, which is the same as their papers. For our proposed method, we set $t=200$. On GradTTS, for Naive Attack, we set $t=0.8$. For SecMI, we set $t=0.6$. Because SecMI is not designed for continuous-time diffusion, we discretize $[0, 1]$ into 1000 steps and then apply SecMI. For the proposed method, we adopt $t=0.3$. We chose $\ell_4$-norm to compute $R_{t,p}$. For other models, we choose the best $t$ because our focus is on the model's robustness. Moreover, as stated later, the difference between the best $t$ and a fixed $t$ is not significant. 

To conduct the experiment on stable diffusion, we download the stable-diffusion-v1-5 from \footnote{\url{https://huggingface.co/runwayml/stable-diffusion-v1-5}}, without any further fine-tuning or any other modification. We select 2500 sample from 600M laion-aesthetics-v2-5plus as the member set, since stable-diffusion-v1-5 is trained on this dataset as mentioned by HuggingFace. We randomly select 2500 images from the COCO2017-val as the hold-out set, since COCO2017-val is one of the official validation set to examine the performance of stable diffusion. For Naive Attach, we set $t=500$. For SecMI, we set $t=100$. For proposed method, we set $t=500$. We also chose $\ell_4$-norm to compute $R_{t,p}$.

\section{More Experimental Results}
\subsection{Robustness on FastDiff and DiffWave}

\cref{fastdiff_diffwave} shows the AUC of different methods at FastDiff and DiffWave model on three datasets. The performance of all three MIA methods is very poor.

\begin{table}[h!]
\centering
\caption{Performance of AUC on FastDiff and DiffWave across three datasets.}
\label{fastdiff_diffwave}
\adjustbox{width=0.9\textwidth}{
\begin{tabular}{c|ccc|ccc}
\toprule
             & \multicolumn{3}{c|}{FastDiff} & \multicolumn{3}{c}{DiffWave} \\ \midrule
Method       & LJSpeech  & VCTK  & LibriTTS  & LJSpeech  & VCTK  & LibriTTS \\ \midrule
NA~\cite{matsumoto2023membership} & 52.6      & 55.1  & 53.7      & 52.7      & 53.8  & 51.2     \\
SecMI~\cite{SecMI}        & 51.6      & 56.3  & 53.7      & 53.2      & 54.3  & 52.4     \\
PIA          & 51.6      & 57.1  & 54.1      & 54.4      & 54.2  & 50.8     \\
PIAN     & 52.4      & 57.0  & 54.6      & 50.0      & 50.5  & 50.7     \\ \bottomrule
\end{tabular}
}
\end{table}

\subsection{Distribution for samples from training set and hold-out set.}

\cref{distribution} shows the $R_{t=0.3,p=4}$ distribution for samples from training set and hold-out set at GradTTS on different datasets of PIAN.

\begin{figure}[h]
\centering
    \begin{subfigure}{0.32\textwidth}
        \centering 
        \includegraphics[width=\textwidth]{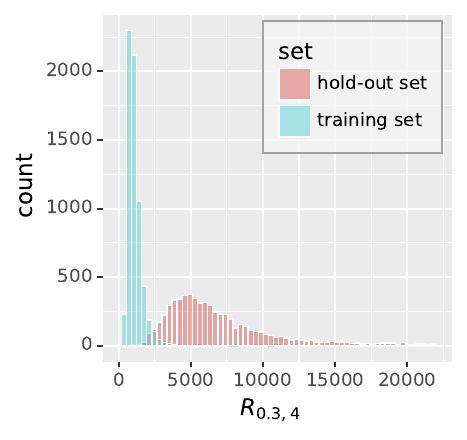} % 1st subfigure: \includegraphics{fig}...
    %   \vspace{-6mm}
    \captionsetup{skip=0pt, width=.95\linewidth}
      \subcaption[]{Distribution on LJSpeech.}
        % \includegraphics[width=0.49\linewidth]{images/gradtts.pdf}
        % \subcaption[subfigcapskip=500pt]{A blue square.}
        % \label{ablation_a}
    \end{subfigure}
    \begin{subfigure}{0.32\textwidth}
    \captionsetup{skip=0pt, width=.95\linewidth}
        \centering \includegraphics[width=\textwidth]{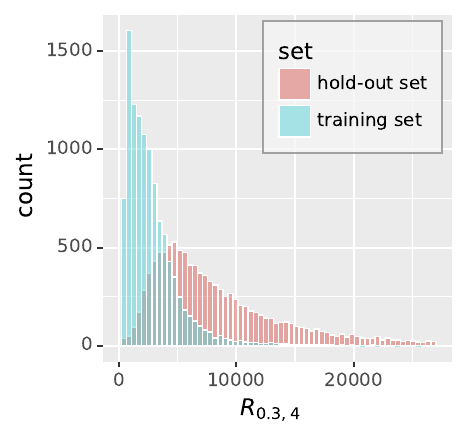}% 2nd subfigure: \includegraphics{fig}...
        % \vspace*{-3mm}
        % \vspace{-6mm}
        \subcaption[]{Distribution on VCTK.}
        % \label{ablation_b}
    \end{subfigure}
    \begin{subfigure}{0.32\textwidth}
    \captionsetup{skip=0pt, width=.95\linewidth}
        \centering \includegraphics[width=\textwidth]{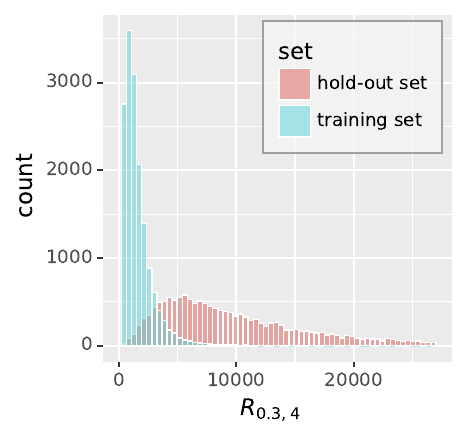}% 2nd subfigure: \includegraphics{fig}...
        % \vspace*{-3mm}
        % \vspace{-6mm}
        \subcaption[]{Distribution on LibriTTS.}
    \end{subfigure}

\caption{$R_{t=0.3,p=4}$ distribution for samples from training set and hold-out set at GradTTS on different datasets of PIAN.}
        % \label{ablation_different}
        \label{distribution}
\vspace{-4mm}
\end{figure}

\subsection{Log-scaled ROC curve}

As suggested by \cite{carlini2022membership}, \cref{log_scaled_ddpm} and \cref{log_scaled_gradtts} display the log-scaled ROC curves. These curves demonstrate that the proposed method outperforms NA and SecMI at most of times.

\begin{figure}[h]
\centering
    \begin{subfigure}{0.45\textwidth}
        \centering 
        \includegraphics[width=\textwidth]{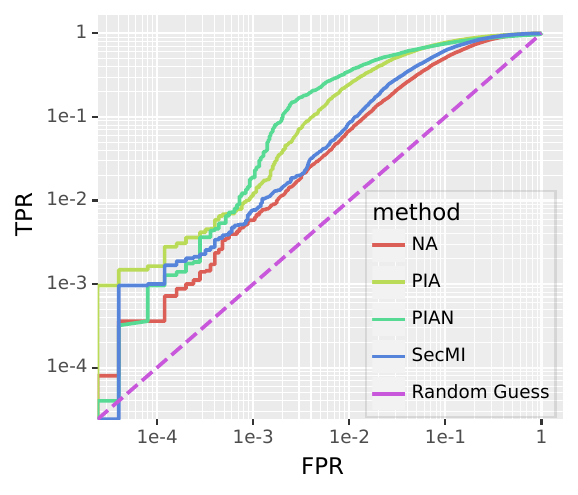} % 1st subfigure: \includegraphics{fig}...
    %   \vspace{-6mm}
    \captionsetup{skip=0pt, width=.95\linewidth}
      \subcaption[]{Log-scaled ROC on CIFAR10.}
        % \includegraphics[width=0.49\linewidth]{images/gradtts.pdf}
        % \subcaption[subfigcapskip=500pt]{A blue square.}
        % \label{ablation_a}
    \end{subfigure}
    % \par\bigskip % maximise vertical space here instead
    \begin{subfigure}{0.45\textwidth}
    \captionsetup{skip=0pt, width=.95\linewidth}
        \centering \includegraphics[width=\textwidth]{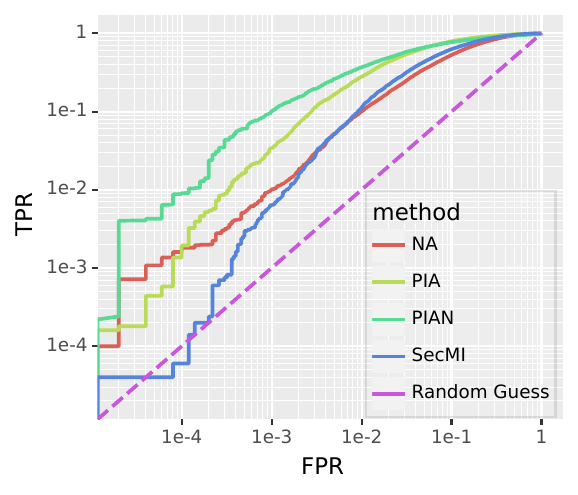}% 2nd subfigure: \includegraphics{fig}...
        % \vspace*{-3mm}
        % \vspace{-6mm}
        \subcaption[]{Log-scaled ROC on TinyImageNet.}
        % \label{ablation_b}
    \end{subfigure}

\caption{The log-scaled ROC at DDPM of different methods on CIFAR10 and TinyImageNet.}
        % \label{ablation_different}
\vspace{-4mm}
\label{log_scaled_ddpm}
\end{figure}

\begin{figure}[t!]
\centering
    \begin{subfigure}{0.45\textwidth}
        \centering 
        \includegraphics[width=\textwidth]{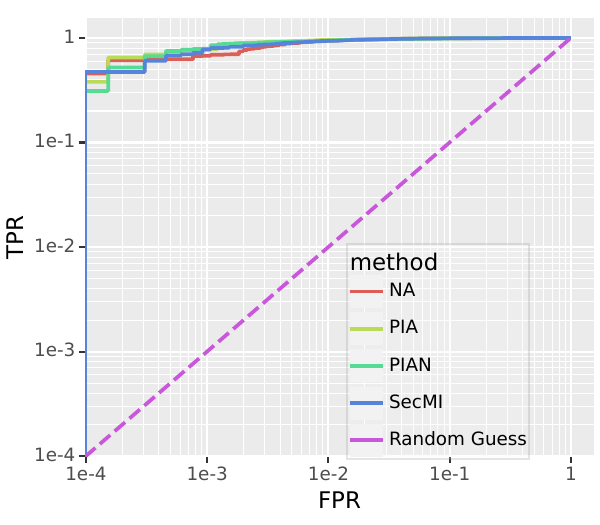} % 1st subfigure: \includegraphics{fig}...
    %   \vspace{-6mm}
    \captionsetup{skip=0pt, width=.95\linewidth}
      \subcaption[]{Log-scaled ROC on LJSpeech.}
        % \includegraphics[width=0.49\linewidth]{images/gradtts.pdf}
        % \subcaption[subfigcapskip=500pt]{A blue square.}
        % \label{ablation_a}
    \end{subfigure}
    \begin{subfigure}{0.45\textwidth}
    \captionsetup{skip=0pt, width=.95\linewidth}
        \centering \includegraphics[width=\textwidth]{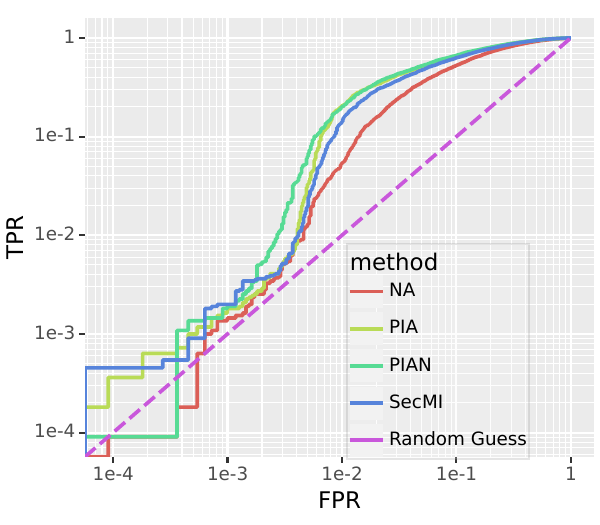}% 2nd subfigure: \includegraphics{fig}...
        % \vspace*{-3mm}
        % \vspace{-6mm}
        \subcaption[]{Log-scaled ROC on VCTK.}
        % \label{ablation_b}
    \end{subfigure}
    
    \begin{subfigure}{0.45\textwidth}
    \captionsetup{skip=0pt, width=.95\linewidth}
        \centering \includegraphics[width=\textwidth]{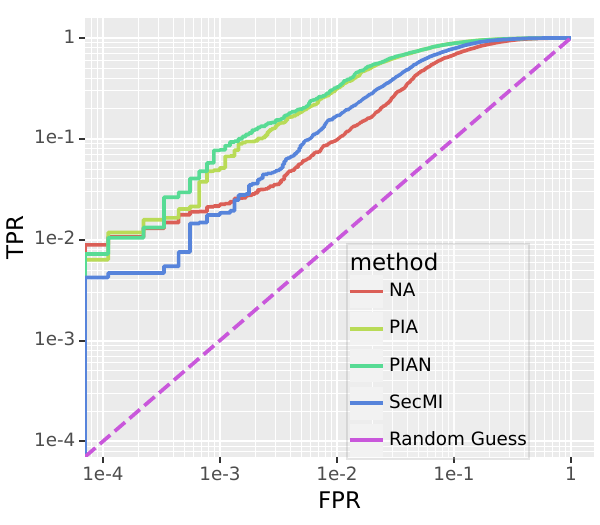}% 2nd subfigure: \includegraphics{fig}...
        % \vspace*{-3mm}
        % \vspace{-6mm}
        \subcaption[]{Log-scaled ROC on LibriTTS.}
        % \label{ablation_b}
    \end{subfigure}
    
\caption{The log-scaled ROC at GradTTS of different methods on LJSpeech, VCTK and LibriTTS.}
        % \label{ablation_different}
\vspace{-4mm}
\label{log_scaled_gradtts}
\end{figure}

\subsection{Visualization of Reconstruction}

Note \cref{final} is equal to the distance between $\boldsymbol{\epsilon}_{\boldsymbol{\theta}}\left( \boldsymbol{x}_{0} , 0\right)$ and the predicted one $\boldsymbol{\epsilon}'=\boldsymbol{\epsilon}_{\boldsymbol{\theta}}\left(\boldsymbol x_t, t \right)$ , where $x_t=\sqrt{\bar{a}_{t}} \boldsymbol{x}_{0}+\sqrt{1-\bar{a}_{t}}  \boldsymbol{\epsilon}_{\boldsymbol{\theta}}\left( \boldsymbol{x}_{0} , 0\right)$. \cref{reconstruction_ddpm} and \cref{reconstructed_nonmembership} show the reconstructed sample $\boldsymbol x'_0 = \frac{\boldsymbol x_t- \sqrt{1-\bar{a}_{t}}\boldsymbol{\epsilon}' }{\sqrt{\bar{a}_{t}}}$ from $\boldsymbol x_t$ using the predicted $\boldsymbol{\epsilon}'$  at DDPM on CIFAR10 of PIAN. The reconstructed samples from $t=100$ are clear for both the training set and the hold-out set. The reconstructed samples from $t=400$ are blurry for both sets. However, for $t=200$, the reconstructed samples are clear for the training set but blurry for the hold-out set.

For GradTTS, we use \cref{ODE_inference} to reconstruct samples from $x_t$. This reconstruction is not rigorous, but we just use it to give a visualization. \cref{gradtts_trainingset} and \cref{gradtts_holdout} show the reconstructed samples on LJSpeech from PIA. The observed pattern is consistent with DDPM.

\begin{figure}
\centering
    \begin{subfigure}{0.6\textwidth}
        \centering 
        \includegraphics[width=\textwidth]{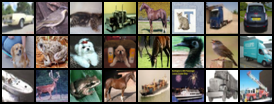} % 1st subfigure: \includegraphics{fig}...
    %   \vspace{-6mm}
    \captionsetup{skip=0pt, width=.95\linewidth}
      \subcaption[]{Samples in training set.}
        % \includegraphics[width=0.49\linewidth]{images/gradtts.pdf}
        % \subcaption[subfigcapskip=500pt]{A blue square.}
        % \label{ablation_a}
    \end{subfigure}
    % \par\bigskip % maximise vertical space here instead
    \begin{subfigure}{0.6\textwidth}
    \captionsetup{skip=0pt, width=.95\linewidth}
        \centering \includegraphics[width=\textwidth]{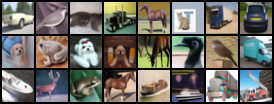}% 2nd subfigure: \includegraphics{fig}...
        % \vspace*{-3mm}
        % \vspace{-6mm}
        \subcaption[]{Samples reconstructed from $t=100$.}
        % \label{ablation_b}
    \end{subfigure}
    \begin{subfigure}{0.6\textwidth}
    \captionsetup{skip=0pt, width=.95\linewidth}
        \centering \includegraphics[width=\textwidth]{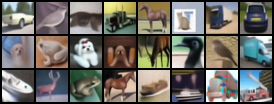}% 2nd subfigure: \includegraphics{fig}...
        % \vspace*{-3mm}
        % \vspace{-6mm}
        \subcaption[]{Samples reconstructed from $t=200$.}
        % \label{ablation_b}
    \end{subfigure}
    \begin{subfigure}{0.6\textwidth}
    \captionsetup{skip=0pt, width=.95\linewidth}
        \centering \includegraphics[width=\textwidth]{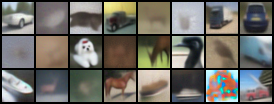}% 2nd subfigure: \includegraphics{fig}...
        % \vspace*{-3mm}
        % \vspace{-6mm}
        \subcaption[]{Samples reconstructed from $t=400$.}
        % \label{ablation_b}
    \end{subfigure}
\caption{Samples in training set and the reconstructed samples at DDPM on CIFAR10 from PIAN.}
        \label{reconstruction_ddpm}
\vspace{-4mm}
\end{figure}

\begin{figure}
\centering
    \begin{subfigure}{0.6\textwidth}
        \centering 
        \includegraphics[width=\textwidth]{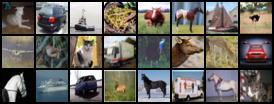} % 1st subfigure: \includegraphics{fig}...
    %   \vspace{-6mm}
    \captionsetup{skip=0pt, width=.95\linewidth}
      \subcaption[]{Samples in hold-out set.}
        % \includegraphics[width=0.49\linewidth]{images/gradtts.pdf}
        % \subcaption[subfigcapskip=500pt]{A blue square.}
        % \label{ablation_a}
    \end{subfigure}
    % \par\bigskip % maximise vertical space here instead
    \begin{subfigure}{0.6\textwidth}
    \captionsetup{skip=0pt, width=.95\linewidth}
        \centering \includegraphics[width=\textwidth]{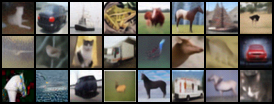}% 2nd subfigure: \includegraphics{fig}...
        % \vspace*{-3mm}
        % \vspace{-6mm}
        \subcaption[]{Samples reconstructed from $t=100$.}
        % \label{ablation_b}
    \end{subfigure}
    \begin{subfigure}{0.6\textwidth}
    \captionsetup{skip=0pt, width=.95\linewidth}
        \centering \includegraphics[width=\textwidth]{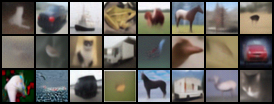}% 2nd subfigure: \includegraphics{fig}...
        % \vspace*{-3mm}
        % \vspace{-6mm}
        \subcaption[]{Samples reconstructed from $t=200$.}
        % \label{ablation_b}
    \end{subfigure}
    \begin{subfigure}{0.6\textwidth}
    \captionsetup{skip=0pt, width=.95\linewidth}
        \centering \includegraphics[width=\textwidth]{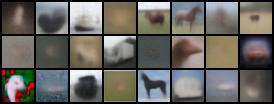}% 2nd subfigure: \includegraphics{fig}...
        % \vspace*{-3mm}
        % \vspace{-6mm}
        \subcaption[]{Samples reconstructed from $t=400$.}
        % \label{ablation_b}
    \end{subfigure}
\caption{Samples in hold-out set and the reconstructed samples at DDPM on CIFAR10 from PIAN.}
        \label{reconstructed_nonmembership}
\vspace{-4mm}
\end{figure}

\begin{figure}
\centering
    \begin{subfigure}{0.8\textwidth}
        \centering 
        \includegraphics[width=\textwidth]{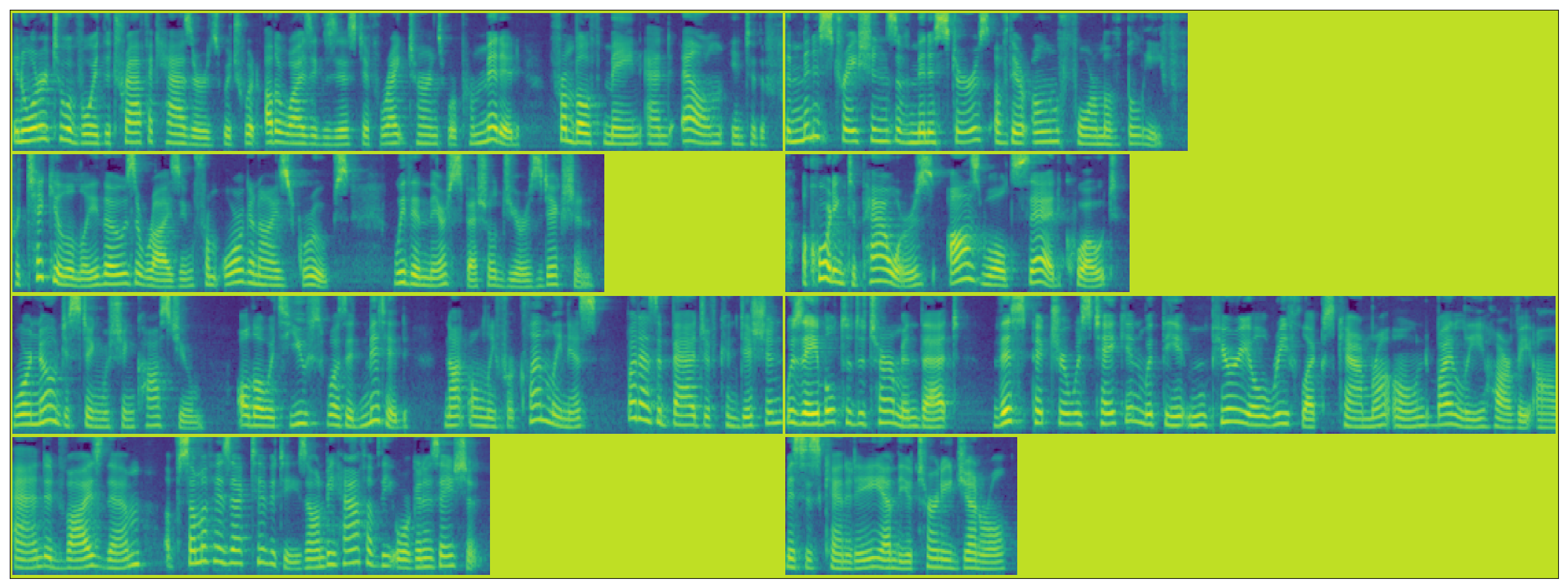} % 1st subfigure: \includegraphics{fig}...
    %   \vspace{-6mm}
    \captionsetup{skip=0pt, width=.95\linewidth}
      \subcaption[]{Samples in training set.}
        % \includegraphics[width=0.49\linewidth]{images/gradtts.pdf}
        % \subcaption[subfigcapskip=500pt]{A blue square.}
        % \label{ablation_a}
    \end{subfigure}
    % \par\bigskip % maximise vertical space here instead
    \begin{subfigure}{0.8\textwidth}
    \captionsetup{skip=0pt, width=.95\linewidth}
        \centering \includegraphics[width=\textwidth]{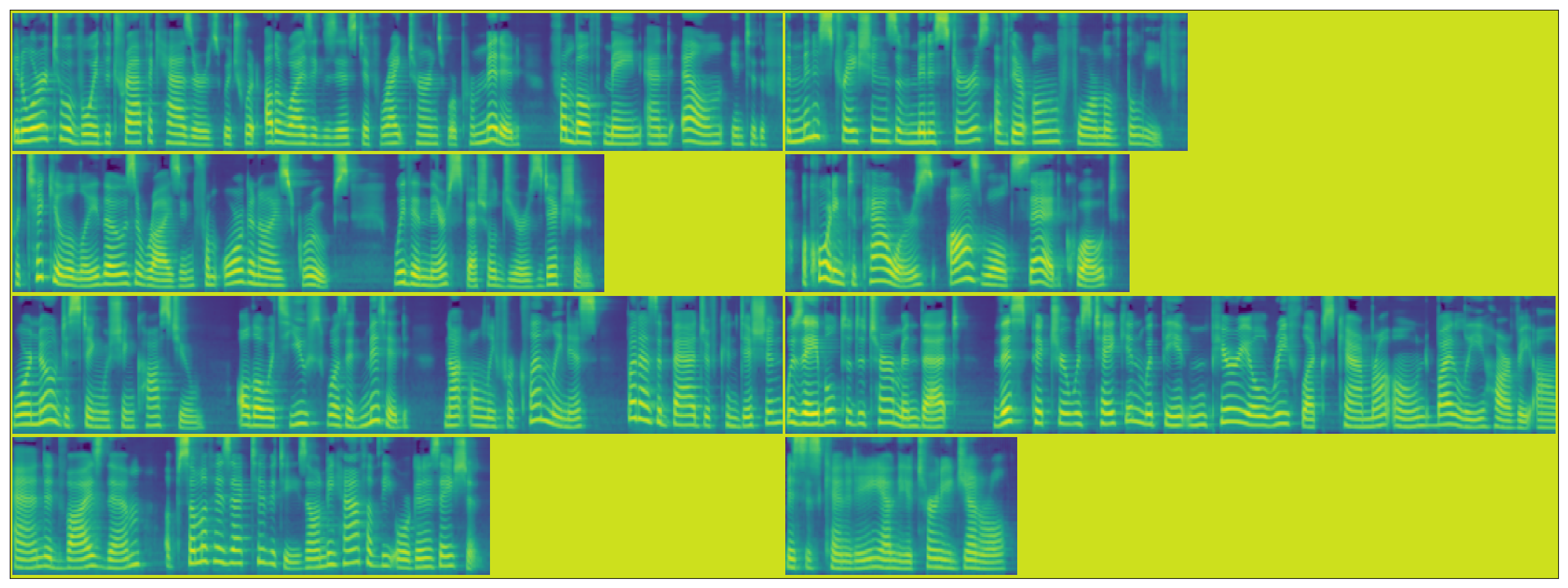}% 2nd subfigure: \includegraphics{fig}...
        % \vspace*{-3mm}
        % \vspace{-6mm}
        \subcaption[]{Samples reconstructed from $t=0.1$.}
        % \label{ablation_b}
    \end{subfigure}
    \begin{subfigure}{0.8\textwidth}
    \captionsetup{skip=0pt, width=.95\linewidth}
        \centering \includegraphics[width=\textwidth]{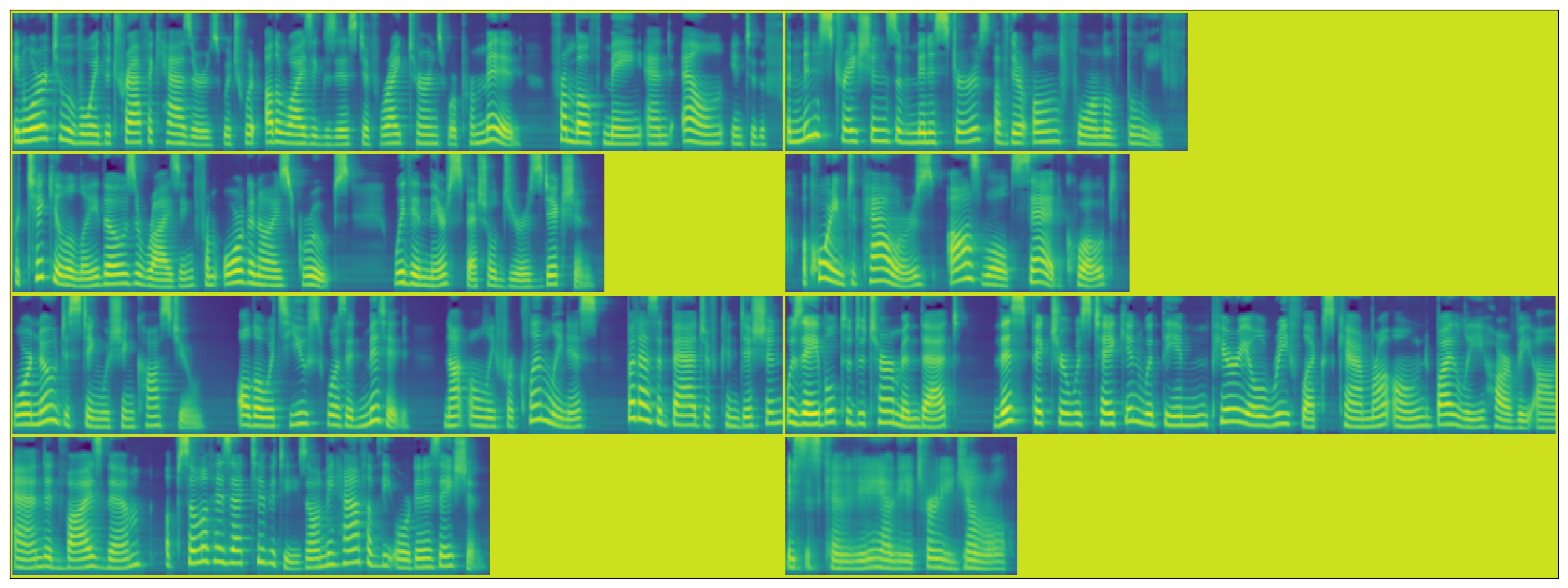}% 2nd subfigure: \includegraphics{fig}...
        % \vspace*{-3mm}
        % \vspace{-6mm}
        \subcaption[]{Samples reconstructed from $t=0.6$.}
        % \label{ablation_b}
    \end{subfigure}
    \begin{subfigure}{0.8\textwidth}
    \captionsetup{skip=0pt, width=.95\linewidth}
        \centering \includegraphics[width=\textwidth]{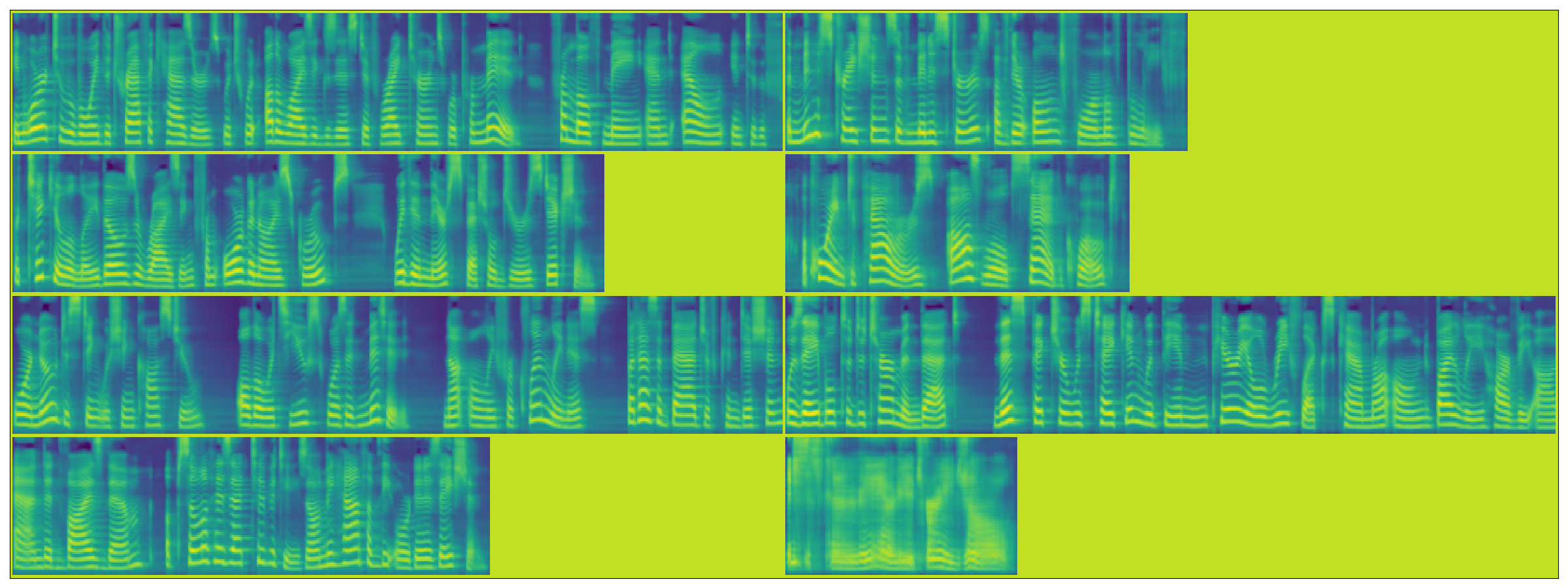}% 2nd subfigure: \includegraphics{fig}...
        % \vspace*{-3mm}
        % \vspace{-6mm}
        \subcaption[]{Samples reconstructed from $t=0.95$.}
        % \label{ablation_b}
    \end{subfigure}
\caption{Samples in training set and the reconstructed samples at GradTTS on LJSpeech from PIA.}
        \label{gradtts_trainingset}
\vspace{-4mm}
\end{figure}

\begin{figure}
\centering
    \begin{subfigure}{0.8\textwidth}
        \centering 
        \includegraphics[width=\textwidth]{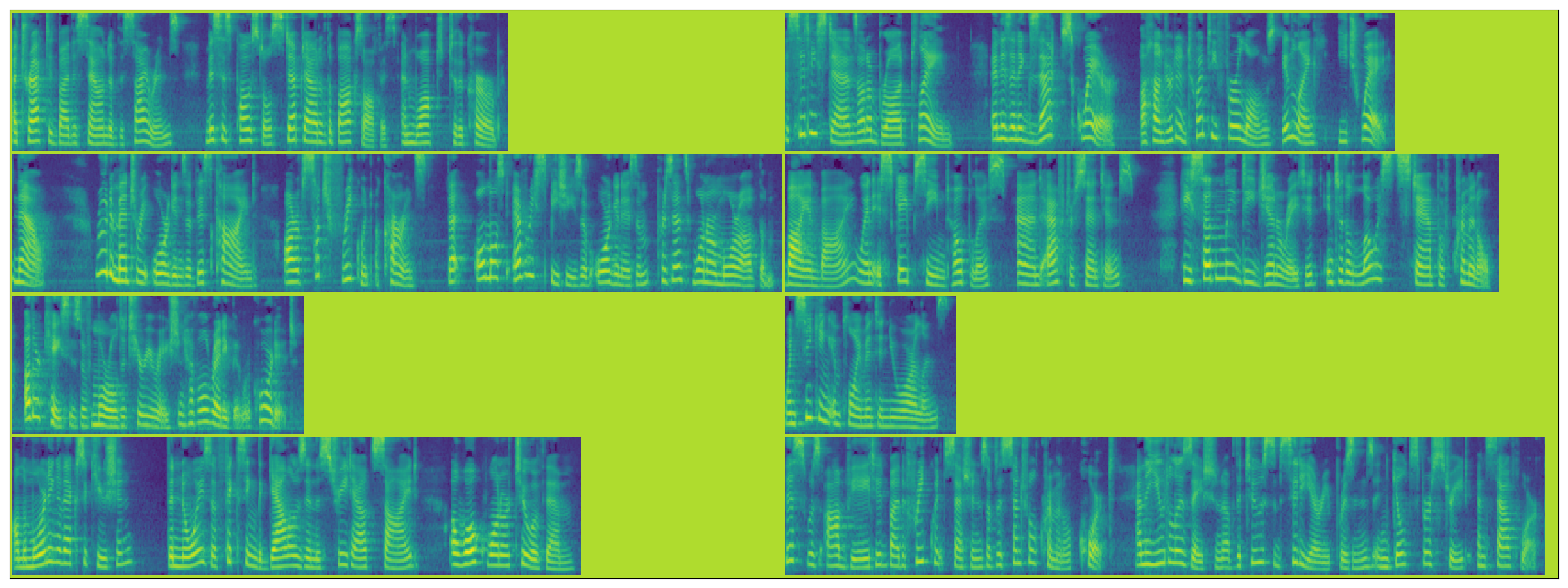} % 1st subfigure: \includegraphics{fig}...
    %   \vspace{-6mm}
    \captionsetup{skip=0pt, width=.95\linewidth}
      \subcaption[]{Samples in hold-out set.}
        % \includegraphics[width=0.49\linewidth]{images/gradtts.pdf}
        % \subcaption[subfigcapskip=500pt]{A blue square.}
        % \label{ablation_a}
    \end{subfigure}
    % \par\bigskip % maximise vertical space here instead
    \begin{subfigure}{0.8\textwidth}
    \captionsetup{skip=0pt, width=.95\linewidth}
        \centering \includegraphics[width=\textwidth]{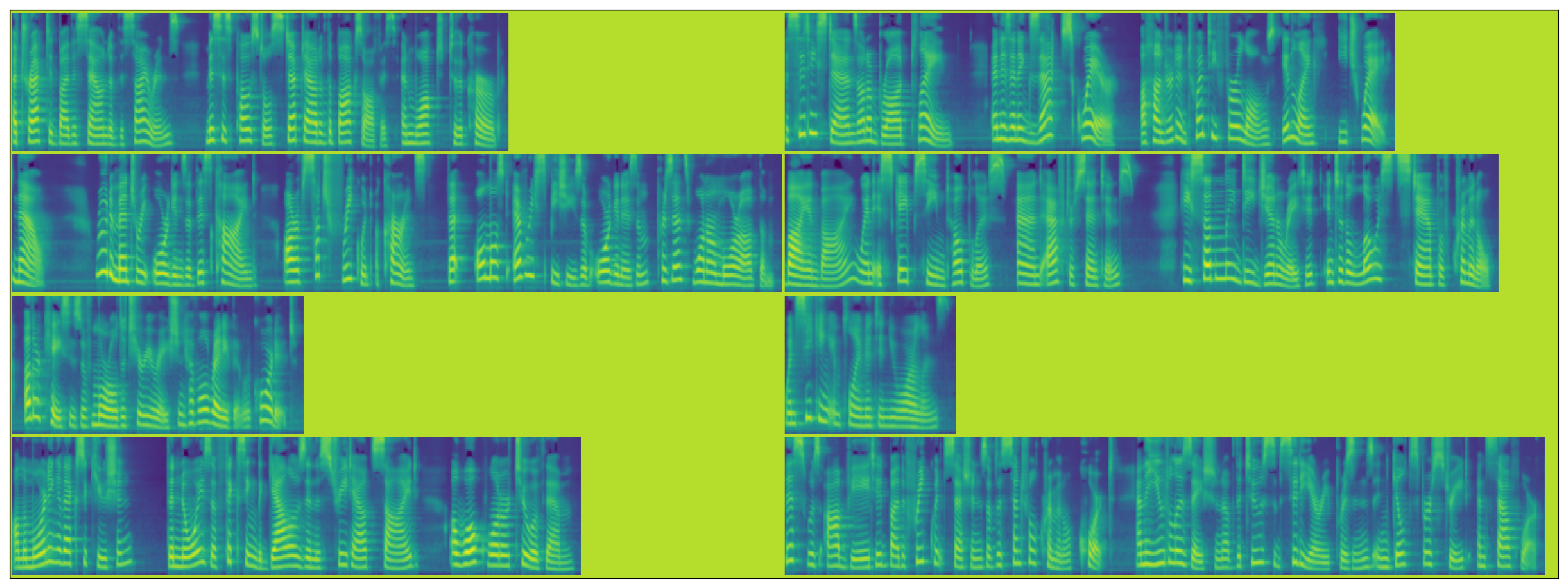}% 2nd subfigure: \includegraphics{fig}...
        % \vspace*{-3mm}
        % \vspace{-6mm}
        \subcaption[]{Samples reconstructed from $t=0.1$.}
        % \label{ablation_b}
    \end{subfigure}
    \begin{subfigure}{0.8\textwidth}
    \captionsetup{skip=0pt, width=.95\linewidth}
        \centering \includegraphics[width=\textwidth]{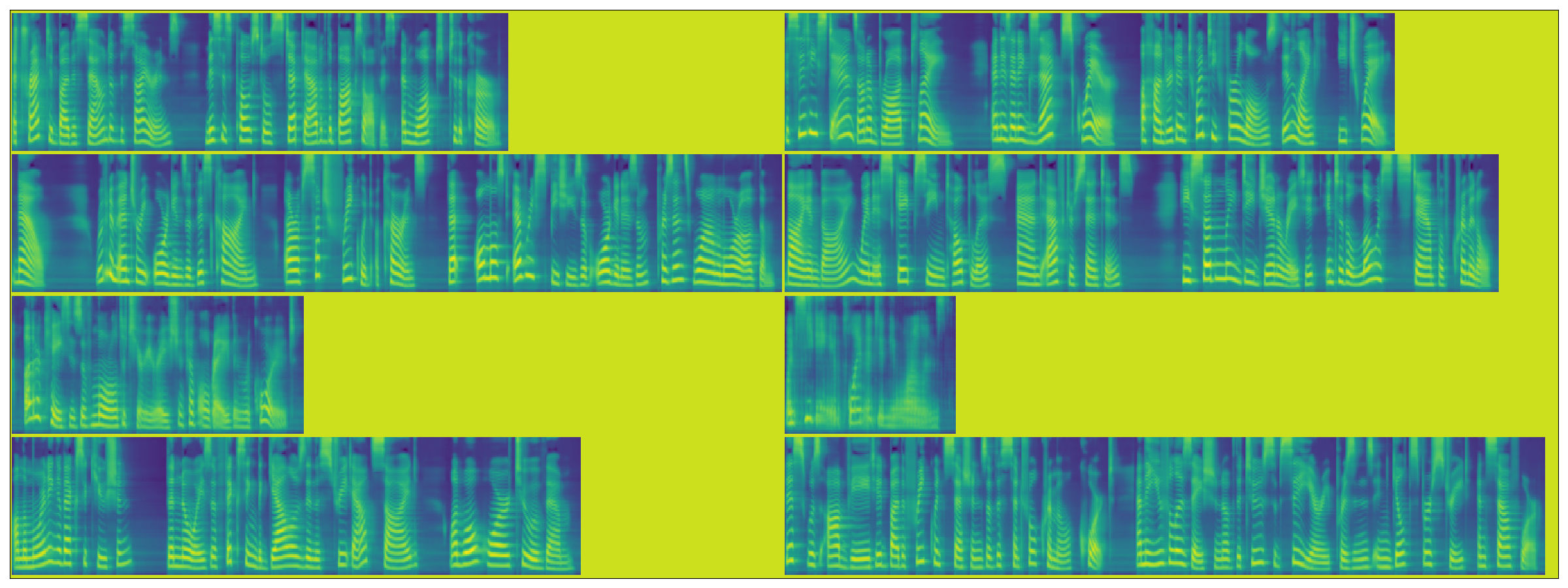}% 2nd subfigure: \includegraphics{fig}...
        % \vspace*{-3mm}
        % \vspace{-6mm}
        \subcaption[]{Samples reconstructed from $t=0.6$.}
        % \label{ablation_b}
    \end{subfigure}
    \begin{subfigure}{0.8\textwidth}
    \captionsetup{skip=0pt, width=.95\linewidth}
        \centering \includegraphics[width=\textwidth]{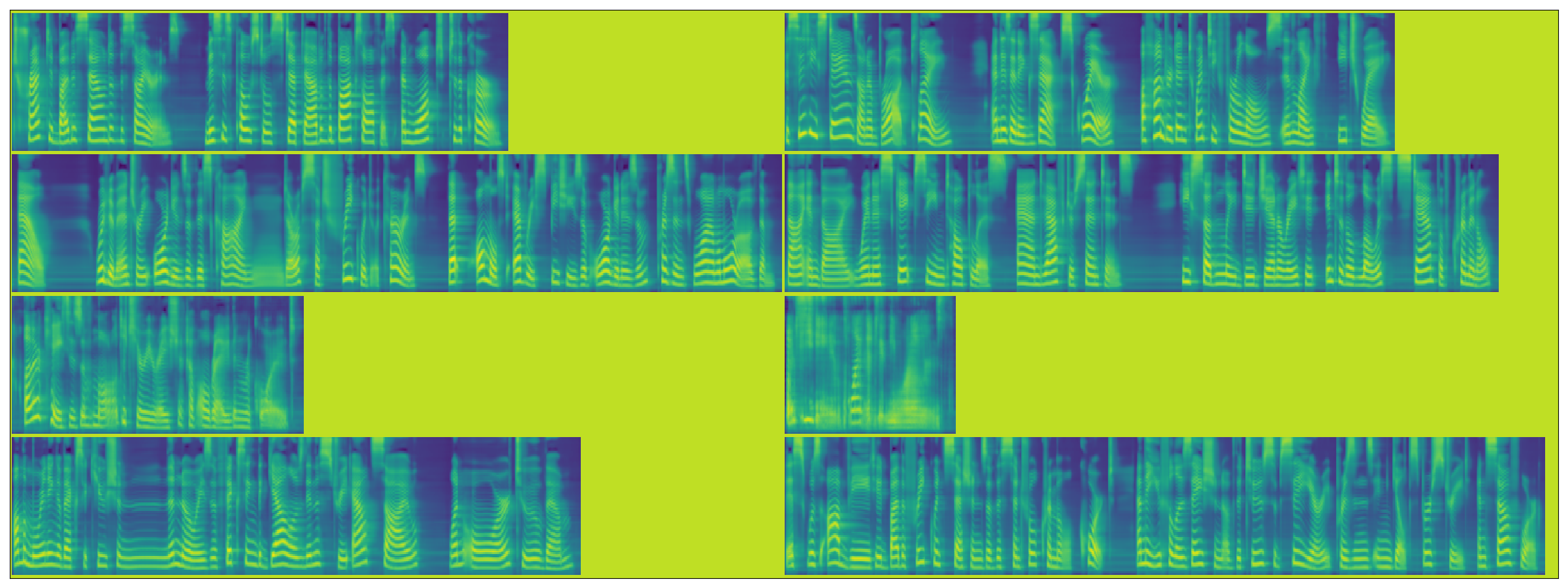}% 2nd subfigure: \includegraphics{fig}...
        % \vspace*{-3mm}
        % \vspace{-6mm}
        \subcaption[]{Samples reconstructed from $t=0.95$.}
        % \label{ablation_b}
    \end{subfigure}
\caption{Samples in hold-out set and the reconstructed samples at GradTTS on LJSpeech from PIA.}
        \label{gradtts_holdout}
\vspace{-4mm}
\end{figure}

% % Please add the following required packages to your document preamble:
% % \usepackage[normalem]{ulem}
% % \useunder{\uline}{\ul}{}
% \begin{table}[]
% \centering
% \caption{Performance of PIA on FastDiff across four datasets was evaluated. The results showed that, with one additional query compared to Naive Attack, PIA could achieve a performance similar to that of SciMI. However, compared to Grad-TTS, the AUC of FastDiff is very poor. It is not much better than random guessing.  (Training step: 2,000,000.) \textbf{(put it into supplementary material)}}
% \label{tab:my-table}
% \begin{tabular}{ccccc}
% \hline
%              &       & LJspeech      & VCTK          & LibriTTS      \\ \hline
% Method       & Query & AUC           & AUC           & AUC           \\ \hline
% Naive Attack & 1     & 53.1          & 54.4          & 53.4          \\ \hline
% SecMI        & 30+2  & {\ul 53.5}    & {\ul 56.6}    & 54.1          \\
% PIA          & 1+1   & 52.1          & 56.4          & {\ul 54.2}    \\
% PIA\_ABS     & 1+1   & \textbf{54.7} & \textbf{57.4} & \textbf{55.3} \\ \hline
% \end{tabular}
% \end{table}

\end{document}